\documentclass[11pt]{article} 

\usepackage[ansinew]{inputenc}
\usepackage{amsmath, amssymb, graphics, amsthm, bbm}
\usepackage{graphicx}
\usepackage{epsfig}
\usepackage{color}
\usepackage{undertilde}
\usepackage{fancyhdr} 
\usepackage{arydshln}
\usepackage{color}
\usepackage{bbold}
\usepackage{slashed}

\newcommand{\ket}[1]{\left| #1 \right\rangle}

\newcommand{\braopket}[3]{\left\langle \vphantom {#1 #2 #3} #1 \hphantom{|} \right| #2 \left| \hphantom{|} \vphantom {#1 #2 #3} #3 \right\rangle}

\newtheorem{theorem}{Theorem}

\newcommand{\m}{\mbox{}}

\makeatletter
\@addtoreset{equation}{section}
\makeatother
\renewcommand{\theequation}{\thesection.\arabic{equation}}

\newcommand{\be}{\begin{equation}}
\newcommand{\ee}{\end{equation}}
\newcommand{\ba}{\begin{eqnarray}}
\newcommand{\ea}{\end{eqnarray}}

\title{{\sf On the Implementation of the}\\
{\sf Canonical Quantum Simplicity Constraint}} 
\author{
{\sf N. Bodendorfer}$^{1,2}$\thanks{{\sf 
norbert.bodendorfer@gravity.fau.de}},
{\sf T. Thiemann}$^{1,3}$\thanks{{\sf 
thomas.thiemann@gravity.fau.de,
tthiemann@perimeterinstitute.ca}},
{\sf A. Thurn}$^1$\thanks{{\sf 
andreas.thurn@gravity.fau.de}}\\
\\
{\sf $^1$ Inst. for Theoretical Physics III, FAU Erlangen -- N\"urnberg,}\\
{\sf Staudtstr. 7, 91058 Erlangen, Germany}\\
\\
{\sf $^2$ Institute for
Gravitation and the Cosmos \& Physics
  Department,}\\
{\sf   Penn State, University Park, PA 16802, U.S.A.}\\
\\
{\sf $^3$ Perimeter Institute for Theoretical Physics,}\\ 
{\sf 31 Caroline Street N, Waterloo, ON N2L 2Y5, Canada}
}
\date{{\small\sf \today}}

\begin{document} 

\maketitle

{\sf

\begin{abstract}
In this paper, we are going to discuss several approaches to solve the quadratic and linear simplicity constraints in the context of the canonical formulations of higher dimensional General Relativity and Supergravity developed in \cite{BTTI,BTTII,BTTIII,BTTIV,BTTVI,BTTVII}.
 Since the canonical quadratic simplicity constraint operators have been shown to be anomalous in any dimension $D \geq 3$ in \cite{BTTIII}, non-standard methods have to be employed to avoid inconsistencies in the quantum theory. We show that one can choose a subset of quadratic simplicity constraint operators which are non-anomalous among themselves and allow for a natural unitary map of the spin networks in the kernel of these simplicity constraint operators to the SU$(2)$-based Ashtekar-Lewandowski Hilbert space in $D=3$. The linear constraint operators on the other hand are non-anomalous by themselves, however their solution space will be shown to differ in $D=3$ from the expected Ashtekar-Lewandowski Hilbert space. We comment on possible strategies to make a connection to the quadratic theory. Also, we comment on the relation of our proposals to existing work in the spin foam literature and how these works could be used in the canonical theory. We emphasise that many ideas developed in this paper are certainly incomplete and should be considered 
as suggestions for possible starting points for more satisfactory treatments in the future.  
\end{abstract}

}

\newpage

\tableofcontents

\vspace{20mm}

\section{Introduction}
In \cite{BTTI,BTTII}, gravity in any dimension $D+1 \geq 3$ has been formulated as a gauge theory of SO$(1,D)$ or of the compact group SO$(D+1)$, irrespective of the spacetime signature. The resulting theory has been obtained on two different routes, a Hamiltonian analysis of the Palatini action making use of the procedure of gauge unfixing\footnote{See \cite{MitraGaugeInvariantReformulation, AnishettyGaugeInvarianceIn, VytheeswaranGaugeUnfixingIn} for original literature on gauge unfixing.}, and on the canonical side by an extension of the ADM phase space. The additional constraints appearing in this formulation, the simplicity constraints, are well known. They constrain bivectors to be simple, i.e. the antisymmetrised product of two vectors. Originally introduced in Plebanski's \cite{PlebanskiOnTheSeparation} formulation of General Relativity as a constrained $BF$ theory in $3+1$ dimensions, they have been generalised to arbitrary dimension in \cite{FreidelBFDescriptionOf} and were considered in the context of Hamiltonian lattice gravity \cite{WaelbroeckAHamiltonianLattice, ZapataTopologicalLatticeGravity}. Moreover, discrete versions of the simplicity constraints are a standard ingredient of the Spin Foam approaches to quantum gravity \cite{BarrettRelativisticSpinNetworks, EngleFlippedSpinfoamVertex, FreidelANewSpin}, see \cite{PerezSpinFoamModels, AlexandrovCriticalOverviewOf} for reviews, and recently were also used in Group Field theory \cite{DePietriBarrett-CraneModelFrom, BenGelounEPRLFKGroupField, BaratinGroupFieldTheory} as well as on a simplicial phase space \cite{DittrichPhaseSpaceDescription, DittrichSimplicityInSimplicial}, where also their algebra was calculated. Two different versions of simplicity constraints are considered in the literature, which are either quadratic or linear in the bivector fields. The quantum operators corresponding to the quadratic simplicity constraints have been found to be anomalous both in the covariant \cite{EngleLoopQuantumGravity} as well as in the canonical picture \cite{WielandComplexAshtekarVariables, BTTIII}. On the covariant side, this lead to one of the major points of critique about the Barrett-Crane model \cite{BarrettRelativisticSpinNetworks}: The anomalous constraints are imposed strongly\footnote{Strongly here means that the constraint operator annihilates physical states, $\hat C \left|\psi\right\rangle = 0 ~ \forall \left| \psi \right\rangle \in \mathcal{H}_{phys}$}, which may imply erroneous elimination of physical degrees of freedom \cite{DiracLecturesOnQuantum}. This triggered the development of the new Spin Foam models \cite{EngleTheLoopQuantum, LivineNewSpinfoamVertex, EngleFlippedSpinfoamVertex, EngleLoopQuantumGravity, FreidelANewSpin, KaminskiSpinFoamsFor}, in which the quadratic simplicity constraints are replaced by linear simplicity constraints. The linear version of the constraints is slightly stronger than the quadratic constraints, since in $3+1$ dimensions the topological solution is absent. The corresponding quantum operators are still anomalous (unless the Immirzi parameter takes the values $\gamma = \pm \sqrt{\zeta}$, where $\zeta$ denotes the internal signature, or $\gamma = \infty$). Therefore, in the new models (parts of) the simplicity constraints are implemented weakly to account for the anomaly. Also, the newly developed U$(N)$ tools \cite{GirelliReconstructingQuantumGeometry, FreidelTheFineStructure, FreidelU(N)CoherentStates} have been recently applied to solve the simplicity constraints \cite{DupuisRevisitingTheSimplicity, DupuisHolomorphicSimplicityConstraints, DupuisHolomorphicLorentzianSimplicity}. 

In this paper, we 
%are not going to import techniques for solving the simplicity constraints which were developed in other contexts, but we 
are first going to take an unbiased look at them from the canonical perspective in the hope of finding new clues for how to implement the constraints correctly. 
Afterwards, we will compare our results to existing approaches from the Spin Foam literature and outline similarities and differences.
We stress that will not arrive at the conclusion that a certain kind of imposition will be the correct one and thus further research, centered around consistency considerations and the classical limit, has to be performed to find a satisfactory treatment for the simplicity constraints.
Of course, in the end an experiment will have to decide which implementation, if any, will be the correct one. Since such experiments are missing up to now, the general guidelines are of course mathematical consistency of the approach, as well as comparison with the classical implementation of the simplicity constraints in $D=3$, where the usual SU$(2)$ Ashtekar variables exist. If a satisfactory implementation in $D=3$ can be constructed, the hope would then be that this procedure has a natural generalisation to higher dimensions. Since parts of the very promising results developed from the Spin Foam literature are restricted to four dimensions, we will restrict ourselves to dimension independent treatments in the main part of this paper.

The paper will be divided into three parts. We will begin with investigating the quadratic simplicity constraint operators which have been shown to be anomalous in \cite{BTTIII}. It will be illustrated that choosing a recoupling scheme for the intertwiner naturally leads to a maximal closing subset of simplicity constraint operators. Next, the solution to this subset will be shown to allow for a natural unitary map to the SU$(2)$ based Ashtekar-Lewandowski Hilbert space in $D=3$ and we will finish the first part with several remarks on this quantisation procedure. In the section \ref{sec:TheLinearSimplicity}, we will analyse the strong implementation of the linear simplicity constraint operators since they are non-anomalous from start. The resulting intertwiner space will be shown to be one-dimensional, which is problematic because this forbids the construction of a natural map to the SU$(2)$ based Ashtekar-Lewandowski Hilbert space.  In contrast to the quadratic case, the linear simplicity constraint operators will be shown to be problematic when acting on edges. We will discuss several possibilities of how to resolve these problems and finally introduce a mixed quantisation, in which the linear simplicity constraints will be substituted by the quadratic constraints plus a constraint ensuring the equality of the normals $N^I$ and $n^I(\pi)$. 
In section \ref{sec:Comparison}, we will compare our results to existing approaches from the Spin Foam literature. Finally, we will give a critical evaluation of our results and conclude in section \ref{sec:Conclusion}.

\section{The Quadratic Simplicity Constraint Operators}

\label{sec:QuadraticSimplicity}

\subsection{A Maximal Closing Subset of Vertex Constraints}

In our companion papers \cite{BTTI,BTTII,BTTIII,BTTIV,BTTVI,BTTVII}, a canonical connection formulation of $(D+1)$-dimensional Lorentzian General Relativity was developed, using an SO$(D+1)$-connection $A_{aIJ}$ and its conjugate momentum $\pi^{aIJ}$ as canonical variables. Here, $a,b,\ldots = 1,\ldots, D$ are spatial tensorial indices and $I,J,\ldots = 0,\ldots, D$ are Lie algebra indices in the fundamental representation. A key input of the construction are the (quadratic) simplicity constraints 
\be
	\pi^{a[IJ} \pi^{b|IJ]} \approx 0 \text{,}
\ee
which enforce, up to a topological sector present in $D=3$, that $\pi^{aIJ} \approx 2 n^{[I} E^{a|J]}$, where $E^{aJ}$ is an SO$(D+1)$ valued vector density, a so called hybrid vielbein, and $n^I$ is the unique (up to sign) normal defined by $n_I E^{aI}=0$. Fixing the time gauge $n^I = (1,0,\ldots, 0)$, one arrives at the ADM (extended phase space) formulation of General Relativity with SO$(D)$ gauge invariance, see \cite{BTTII} for details. The second class constraints which normally arise as stability conditions on the simplicity constraints are absent in our connection formulation, since they can be explicitly removed by the process of gauge unfixing after performing the Dirac analysis, see \cite{BTTII}. Essentially, they are gauge fixing conditions for the gauge transformations generated by the simplicity constraint, which change a certain part of the torsion of the $A_{aIJ}$. The square of this part of the torsion is included in a respective decomposition of the Palatini action and thus results in the second class partner for the simplicity constraint \cite{BTTII}.

A quantisation of the simplicity constraint using loop quantum gravity methods results in a complicated operator, since $\pi^{aIJ}$ becomes a flux operator which acts as the sum of all right invariant vector fields associated to the different edges at a vertex. In order to facilitate the treatment of this quantum constraint, it has been shown in \cite{BTTIII} that the necessary and sufficient building blocks of the quadratic simplicity constraint operator acting on a vertex $v$ are given by
\be 
R^e_{[IJ}R^{e'}_{KL]}f_{\gamma}=0 \hspace{5mm} \forall e,e' \in \{e'' \in E(\gamma); v = b(e'')\}  \text{,} \label{eq:allsimplicities}
\ee
where $R^e_{IJ}$ is the right invariant vector field associated to the edge $e$, $f_{\gamma}$ is the a cylindrical function defined on an adapted graph $\gamma$, e.g. a spin network, $v$ is a vertex of $\gamma$, $E(\gamma)$ is the set of edges of $\gamma$ and $b(e)$ denotes the beginning of the edge $e$. The orientations of all edges are chosen such that they are outgoing of $v$. We note that these are exactly the off-diagonal simplicity constraints familiar from spin foam models, see e.g. \cite{FreidelBFDescriptionOf, EngleLoopQuantumGravity}.

Since not all of these building blocks commute with each other, i.e. the ones sharing exactly one edge, we will have to resort to a non-standard procedure in order to avoid an anomaly in the quantum theory. The strong imposition of the above constraints, leading to the Barrett-Crane intertwiner \cite{BarrettRelativisticSpinNetworks}, was discussed in \cite{FreidelBFDescriptionOf}. A master constraint formulation of the vertex simplicity constraint operator was proposed in \cite{BTTIII}, however apart from providing a precise definition of the problem, this approach has not lead to a concrete solution up to now. 

In this paper, we are going to explore a different strategy for implementing the quadratic vertex simplicity constraint operators which is guided by two natural requirements:
\begin{enumerate}
\item The imposition of the constraints should be non-anomalous.
\item The imposition of the simplicity constraint operator in $D=3$ should, at least on the kinematical level, lead to the same Hilbert space as the quantisation of the classical theory without a simplicity constraint. More precisely, there should exist a natural unitary map from the solution space of the quadratic simplicity constraint operators $\mathcal{H}_\text{simple}$ to the Ashtekar-Lewandowski Hilbert space $\mathcal{H}_\text{AL}$ in $D=3$.
\end{enumerate}

The concept of gauge unfixing \cite{MitraGaugeInvariantReformulation, AnishettyGaugeInvarianceIn, VytheeswaranGaugeUnfixingIn} which was successfully used in order to derive the classical connection formulation of General Relativity \cite{BTTI,BTTII} used in this paper was originally developed in the context of anomalous gauge theory, where it was observed that first class constraints can turn into second class constraints after quantisation  \cite{MitraGaugeInvariantReformulationAnomalous, JackiwVectorMesonMass, LottDegreesOfFreedom, FaddeevOperatorAnomalyFor, RajaramanHamiltonianFormulationOf}. This is however precisely what is happening in our case: The classically Abelian simplicity constraints become a set of non-commuting operators due to the regularisation procedure used for the fluxes.  The natural question arising is thus: How does a set of maximally commuting vertex simplicity constraint operators look like?

\begin{theorem} \label{thm}
Given a $N$-valent vertex $v \in \gamma$, the set
\begin{eqnarray}
   \epsilon_{IJKL \overline{M}} R_{e_1}^{IJ}  R_{e_1}^{KL} = \ldots =  \epsilon_{IJKL \overline{M}} R_{e_N}^{IJ}  R_{e_N}^{KL} =0 \label{eq:diagonalsimplicities}\\
  \cdashline{1-2} \nonumber \\
    \epsilon_{IJKL \overline{M}} \left( R_{e_1}^{IJ} +R_{e_2}^{IJ} \right)  \left(R_{e_1}^{KL} + R_{e_2}^{KL}  \right)= 0 \nonumber \\
  \epsilon_{IJKL \overline{M}} \left( R_{e_1}^{IJ} +R_{e_2}^{IJ} +R_{e_3}^{IJ} \right)  \left(R_{e_1}^{KL} + R_{e_2}^{KL} +R_{e_3}^{KL}  \right)= 0 \nonumber \\
  \ldots \nonumber \\
   \epsilon_{IJKL \overline{M}} \left( R_{e_1}^{IJ} + \ldots +R_{e_{N-2}}^{IJ} \right)  \left(R_{e_1}^{KL} +\ldots+R_{e_{N-2}}^{KL}  \right)= 0 \label{eq:offdiagonalsimplicities}
\end{eqnarray} 
generates a closed algebra of vertex simplicity constraint operators. Under the assumption that no linear combinations with different multi-indices $\overline{M} = M_1 M_2 \ldots M_{D-3}$ are allowed\hspace{2pt}\footnote{A superposition of different multi-indices seems to be highly unnatural since an anomaly with the Gau{\ss} constraint has to be expected. We are however currently not aware of a proof which excludes this possibility from the viewpoint of a maximal closing set.}, the set is maximal in the sense that adding new vertex constraint operators spoils closure.
\end{theorem}

\begin{proof}
Closure can be checked by explicit calculation. In order to understand why the calculation works, recall that right invariant vector fields generate the Lie algebra so$(D+1)$ as \cite{BTTIII}
\be 
\left[R^e_{IJ},R^{e'}_{KL}\right] =  \frac{1}{2} \delta_{e,e'} \left(\eta_{JK} R^e_{IL} + \eta_{IL} R^e_{JK} -\eta_{IK} R^e_{JL}-\eta_{JL} R^e_{IK}\right) 
\ee
and thus infinitesimal rotations. The commutativity of (\ref{eq:diagonalsimplicities}) has been discussed in \cite{BTTIII}. Further, we see that every element of (\ref{eq:offdiagonalsimplicities}) operates on  (\ref{eq:diagonalsimplicities}) as an infinitesimal rotation. The same is also true for the elements in  (\ref{eq:offdiagonalsimplicities}): Taking the ordering from above, every constraint operates as an infinitesimal rotation on all constraints prior in the list. Since the commutator is antisymmetric in the exchange of its arguments, closure, i.e. commutativity up to constraints, of  (\ref{eq:offdiagonalsimplicities}) follows. 

To prove maximality of the set we will show that, having chosen a subset of simplicity constraints as given in (\ref{eq:diagonalsimplicities}) and (\ref{eq:offdiagonalsimplicities}), adding any other linear combination of the building blocks (\ref{eq:allsimplicities}) spoils the closure of the algebra. To this end, we make the most general Ansatz
\be
\sum_{1\leq i<j <N} \alpha_{ij} ~ \epsilon_{IJKL\overline{M}} R^{IJ}_i R^{KL}_j
\label{eq:Ansatz}
\ee
for an $N$-valent vertex. Note that the diagonal terms $(i=j)$ are proportional to (\ref{eq:diagonalsimplicities}) and therefore do not have to be taken into account in the above sum, and that $R_{N} = \sum_{i = 1}^{N-1} R_i$ can be dropped due to gauge invariance. Moreover, $\alpha_{ij}$ can be chosen such that for fixed $j'$ not all $\alpha_{ij'}$ ($i < j'$) are equal. Otherwise, with $\alpha_{ij'} := \alpha_{j'}$ we find the term $\alpha_{j'} \epsilon_{IJKL\overline{M}}R^{IJ}_{1...(j'-1)} R^{KL}_{j'}$ in the sum, which can be expressed as a linear combination of (\ref{eq:diagonalsimplicities}) and (\ref{eq:offdiagonalsimplicities}) and therefore can be dropped. Consider
\ba
&\m& \left[  \epsilon_{IJKL\overline{M}} R^{IJ}_{12} R^{KL}_{12},  \epsilon_{ABCD\overline{E}} \left(\alpha_{13} R^{AB}_1 R^{CD}_3 + \alpha_{23} R^{AB}_2 R^{CD}_3 + ... \right) \right] \nonumber \\
&\approx& \sum_{j = 3}^{N-1} 2 \alpha_{1j} ~ \epsilon_{IJKL\overline{M}} R^{IJ}_{2} f^{KL~AB}\m_{MN} R^{MN}_{1} \epsilon_{ABCD\overline{E}} R^{CD}_j \nonumber \\
&+&  \sum_{j = 3}^{N-1} 2 \alpha_{2j} ~ \epsilon_{IJKL\overline{M}} R^{IJ}_{1} f^{KL~AB}\m_{MN} R^{MN}_{2} \epsilon_{ABCD\overline{E}} R^{CD}_j \nonumber \\
&\approx&  \sum_{j = 3}^{N-1} 2 (\alpha_{1j} - \alpha_{2j}) ~ \epsilon_{IJKL\overline{M}} R^{IJ}_{2} f^{KL~AB}\m_{MN} R^{MN}_{1} \epsilon_{ABCD\overline{E}} R^{CD}_j \text{,}
\label{eq:thm1}
\ea
where we dropped terms proportional to (\ref{eq:diagonalsimplicities}) in the first and in the second step. For a closing algebra, the right hand side of (\ref{eq:thm1}) necessarily has to be proportional to (a linear combination of) simplicity building blocks (\ref{eq:allsimplicities}). Terms containing $R_{j}$ ($j \geq 3$) have to vanish separately (In general, one could make use of gauge invariance to ``mix" the contributions of different $R_j$. However, in the case at hand this will produce terms containing $R_N$, which do not vanish if the contributions of different $R_j$s did not already vanish separately).

We start with the case $D=3$. The summands on the right hand sides of (\ref{eq:thm1}) are proportional to
\be
\delta_{IJK}^{ABC} (R_j)_{AB} (R_{2})^{IJ} (R_1)^K\m_C\text{,}
\ee
where we used the notation $\delta^{I_1...I_n}_{J_1...J_n} := n! \; \delta^{I_1}_{[J_1} \delta^{I_2}_{J_2} ... \delta^{I_n}_{J_n]}$. To show that this expression can not be rewritten as  a linear combination of the of building blocks (\ref{eq:allsimplicities}) we antisymmetrise the indices $[ABIJ]$, $[ABKC]$ and $[IJKC]$ and find in each case that the result is zero. 

For $D>3$, the summands are proportional to 
\be
\delta_{IJK\overline{M}}^{ABC\overline{E}} (R_j)_{AB} (R_{2})^{IJ} (R_1)^K\m_C \text{.}
\ee
Whatever multi-index $\overline{E}$ we might have chosen in the Ansatz (\ref{eq:Ansatz}), we can always restrict attention to those simplicity constraints in the maximal set which have the same multi-index $\overline{M} = \overline{E}$. Then, the same calculation as in the case of $D=3$ shows that the antisymmetrisations of the indices $[ABIJ]$, $[ABKC]$ and $[IJKC]$ vanish.

Therefore, the only possibilities are $(a)$ the trivial solution $\alpha_{1j} = \alpha_{2j} = 0$ or $(b)$ $\alpha_{1j} = \alpha_{2j}(\neq 0)$, which implies that the terms on the right hand side of (\ref{eq:thm1}) are a rotated version of $\epsilon_{IJKL\overline{M}} R^{IJ}_{1} R^{KL}_{2}$. The second option $(b)$ is, for $j=3$, excluded by our choice of $\alpha_{ij}$ and we must have $\alpha_{13} = \alpha_{23} = 0$. Next, consider $j=4$ and suppose we have $\alpha_{14} = \alpha_{24} := \alpha' \neq 0$. Then, we can define $\alpha'_{34} := \alpha_{34} - \alpha'$ and find the terms $\alpha' \epsilon_{IJKL\overline{M}} R^{IJ}_{123} R^{KL}_{4} + \alpha'_{34} \epsilon_{IJKL\overline{M}} R^{IJ}_{3} R^{KL}_{4}$ in (\ref{eq:Ansatz}). The first term again is already in the chosen set, which implies we can set $\alpha_{14} = \alpha_{24} = 0$ w.l.o.g. by changing $\alpha_{34} \rightarrow \alpha'_{34}$ (We will drop the prime in the following). This immediately generalises to $j>4$, and we have w.l.o.g. $\alpha_{1j} = \alpha_{2j} = 0$ ($3 \leq j < N$).

Suppose we have calculated the commutators of $\epsilon_{IJKL\overline{M}} R^{IJ}_{1...i} R^{KL}_{1...i}$ ($i = 2,...,n$) with (\ref{eq:Ansatz}) and found that for closure, we need $\alpha_{ij} = 0$ for $1\leq i \leq n$ and $i<j<N$. Then,
\ba
&\m& \left[  \epsilon_{IJKL\overline{M}} R^{IJ}_{1...(n+1)} R^{KL}_{1...(n+1)},  \epsilon_{ABCD\overline{E}} \left(\sum_{j=n+2}^{N-1}\alpha_{(n+1)j} R^{AB}_{(n+1)} R^{CD}_j + ... \right) \right] \approx \nonumber \\
&\approx& \sum_{j = (n+2)}^{N-1} 2 \alpha_{(n+1)j} ~ \epsilon_{IJKL\overline{M}} R^{IJ}_{1...n} f^{KL~AB}\m_{MN} R^{MN}_{(n+1)} \epsilon_{ABCD\overline{E}} R^{CD}_j \text{,}
\label{eq:thm2}
\ea
which, by the reasoning above, again is not a linear combination of any simplicity building blocks for any choice of $\alpha_{(n+1)j}$, and therefore only the trivial solution $\alpha_{(n+1)j} = 0$ ($n+1 < j < N$) leads to closure of the algebra.
\end{proof}

\subsection{The Solution Space of the Maximal Closing Subset}

In order to interpret this set of constraints recall from \cite{BTTIII} that the constraints in (\ref{eq:diagonalsimplicities}) are the same as the diagonal simplicity constraints acting on edges of $\gamma$ and can be solved by demanding the edge representations to be simple. The remaining constraints  (\ref{eq:offdiagonalsimplicities}) can be interpreted as specifying a recoupling scheme for the intertwiner $\iota$ at $v$: Couple the representations on $e_1$ and $e_2$, then couple this representation to $e_3$, and so forth, see fig. \ref{fig:recoupling}. We call the intermediate virtual edges $e_{12}$, $e_{123}$, $\ldots$ and denote the highest weights of the representations thereon by $\vec{\Lambda}_{12}, \vec{\Lambda}_{123}, \ldots$ Since we can use gauge invariance at all the intermediate intertwiners in the recoupling scheme, e.g.,  $R_{e_1}+R_{e_2}=R_{e_{12}}$, we have
\be
 \epsilon_{IJKL \overline{M}} \left( R_{e_1}^{IJ} +R_{e_2}^{IJ} \right)  \left(R_{e_1}^{KL} + R_{e_2}^{KL} \right) =  \epsilon_{IJKL \overline{M}} R_{e_{12}}^{IJ} R_{e_{12}}^{KL} = 0  
 \ee
and thus that the representation on $e_{12}$ has to be simple, i.e.
\be
\vec{\Lambda}_{12} = (\lambda_{12},0,...,0) \hspace{5mm} \lambda_{12} = 0,1,2,...
\ee
Using the same procedure, all intermediate representations are required to be simple and the intertwiner is labeled by $N-3$ ``spins'' $\lambda_i \in \mathbb{N}_0$. We call an intertwiner where all internal lines are labeled with simple representations simple. 

\begin{figure}[h]
\centering
\includegraphics[trim = 0mm 170mm 0mm 0mm, clip, scale=0.5]{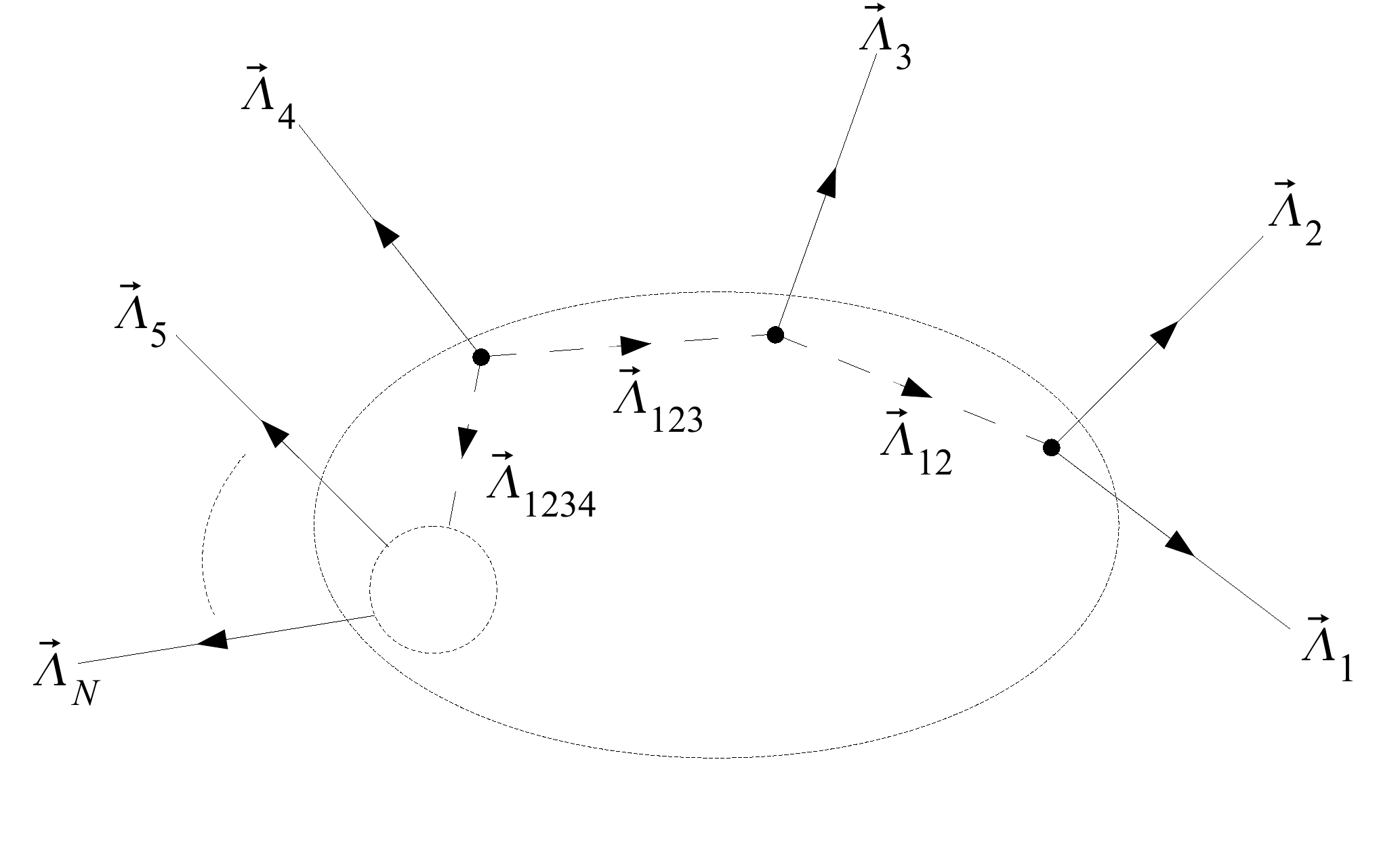}
\caption{Recoupling scheme corresponding to the subset of quadratic vertex simplicity constraint operators (\ref{eq:offdiagonalsimplicities}).}
\label{fig:recoupling}
\end{figure}

Denote by $\mathcal{I}_N^{\text{SU}(2)}$ the set of $N$-valent SU$(2)$ intertwiners and by $\mathcal{I}_{s,N}^{\text{Spin}(D+1)}$ the set of simple $N$-valent \mbox{Spin$(D+1)$} intertwiners. Recalling that an $N$-valent $\text{SU}(2)$ intertwiner can be expressed in the same recoupling basis and calling the intermediate spins $j_i$, we see that the map
\begin{eqnarray}
F : \mathcal{I}_{s,N}^{\text{Spin}(D+1)} &\rightarrow& \mathcal{I}_N^{\text{SU}(2)} \nonumber \\
\frac{1}{2} \lambda_i &\mapsto & j_i
\end{eqnarray}
is unitary (with respect the scalar products induced by the respective Ashtekar-Lewandowski measures, see \cite{BTTIII}). The motivation for the factor $1/2$ comes from the fact that $\vec{\Lambda}=(1,0)$ in $D=3$ corresponds to the familiar $j_+ = j_- = 1/2$ and the area spacings of the SO$(4)$ and the SU$(2)$ based theories agree using this identification, cf \cite{BTTIII}.

\subsection{Remarks}
\begin{enumerate}
\item Since the choice of the maximal closing subset of the simplicity constraint operators is arbitrary, no recoupling basis is preferred a priori. On the SU$(2)$ level, a change in the recoupling scheme amounts to a change of basis in the intertwiner space and therefore poses no problems. On the level of simple Spin$(D+1)$ representations however, a choice in the recoupling scheme affects the property ``simple'', since the non-commutativity of constraint operators belonging to different recoupling schemes means that kinematical states cannot have the property simple in both schemes. 

\item There exist recoupling schemes which are not included in the above procedure, e.g., take $N=6$ and the constraints $\epsilon R_{12} R_{12} = \epsilon R_{34} R_{34} =  \epsilon R_{56} R_{56} = 0$ and couple the three resulting simple representations. The theorem should however generalise to those additional recoupling schemes.

\item It is doubtful if the action of the Hamiltonian constraint leaves the space of simple intertwiners in a certain recoupling scheme invariant. To avoid this problem, one could use a projector on the space of simple intertwiners in a certain recoupling scheme to restrict the Hamiltonian constraint on this subspace and average later on over the different recoupling schemes if they turn out to yield different results. The possible drawbacks of such a procedure are however presently unclear to the authors and we refer to further research.
The construction of such a projector can be seen as a quantum analogue of the gauge unfixing process familiar from our companion paper \cite{BTTII}. A possible strategy to find a Hamiltonian constraint operator which leaves the solution space of a first class subset invariant is to construct a gauge unfixing projector which adds vertex simplicity constraints which are not in the first class subset to the Hamiltonian constraint such that it commutes with the first class subset.

\item It would be interesting to check whether the dropped constraints are automatically solved 
in the weak operator topology (matrix elements with respect to solutions to the 
maximal subset).

\item The imposition of the constraints can be stated  as the search for the joint kernel
of a maximal set of commuting generalised area operators
\be
 \text{Ar}_{\overline{M}}[S] := \sum_{U \in \mathcal{U}} \sqrt{\frac{1}{4} \epsilon_{IJKL\overline{M}} \pi^{IJ}(S_U)\pi^{KL}(S_U)|} \text{.}
\ee
Notice, however, that for  $D>3$ these generalised area operators, just as the simplicity constraints, are not gauge invariant while in $D=3$ they are 

\item In $D=3$ we have the following special situation:\\
We have two classically equivalent extensions of the ADM phase at our disposal whose 
respective symplectic reduction reproduces the ADM phase space. One of them is the 
Ashtekar-Barbero-Immirzi connection formulation in terms of the gauge group SU$(2)$
with additional SU$(2)$ Gau{\ss} constraint next to spatial diffeomorphism and Hamiltonian constraint,
and the other is our connection formulation in terms of SO$(4)$ with additional SO$(4)$ 
Gau{\ss} constraint and simplicity constraint. Both formulations are 
classically completely equivalent and thus one should expect that also the quantum theories
are equivalent in the sense that they have the same semiclassical limit. Let us ask a stronger
condition, namely that the joint kernel of SO$(4)$ Gau{\ss} and simplicity constraint of the SO$(4)$
theory is unitarily equivalent to the kernel of the SU$(2)$ Gau{\ss} constraint of the SU$(2)$ theory.
To investigate this first from the classical perspective, we split the SO$(4)$ connection and 
its conjugate momentum $(A^{IJ},\pi_{IJ})$  into self-dual and anti-selfdual parts $A_\pm^j,\pi^\pm_j)$ which then turn
out to be conjugate pairs again. It is easy to see that the 
SO$(4)$ Gau{\ss} constraint $G_{IJ}$ splits into two SU$(2)$ Gau{\ss} constraints $G^\pm_j$, one involving only 
self-dual variables and the other only anti-selfdual ones which therefore mutually commute
as one would expect. The SO$(4)$ Gau{\ss} constraint now asks for separate SU$(2)$ gauge invariance
for these two sectors.  Thus a quantisation in the Ashtekar-Isham-Lewandowski
representation would yield a kinematical Hilbert space with an orthonormal basis $T^+_{s_+}\otimes 
T^-_{s_-}$ where $S_\pm$ are usual SU$(2)$ invariant spin networks. 
The simplicity constraint, which in $D=3$ is Gau{\ss} invariant and can be imposed after 
solving the Gau{\ss} constraint, from classical perspective asks that the double density
inverse metrics $q^{ab}_\pm=\pi^a_{j\pm} \pi^b_{k\pm} \delta^{jk}$ are identical. This 
is classically equivalent to the statement that corresponding area functions
${\rm Ar}_\pm(S)$ are identical  for every $S$. The corresponding statement in the 
quantum theory is, however, again anomalous because it is well known that area operators 
do not commute with each other. On the other hand, neglecting this complication for a 
moment, it is clear that the quantum constraint can only be satisfied on vectors 
of the form $T^+_{s_+}\otimes T^-_{s_-}$  for all $S$ if $s_+,s_-$ share the same graph
and SU$(2)$ representations on the edges because if $S$ cuts a single edge transversally 
then the area operator is diagonal with an eigenvalue $\propto\sqrt{j(j+1)}$ and we can 
always arrange such an intersection situation by choosing suitable $S$. By a similar
argument one can show that the intertwiners at the edges have to be the same. But 
this is only a sufficient condition because in a sense there are too many quantum simplicity
constraints due to the anomaly. However, the discussion suggests that the joint kernel
of both SO$(4)$ and simplicity constraint is the closed linear span of vectors of the form
$T^+_s\otimes T^-_s$ for the {\it same} spin network $s=s_+=s_-$. The desired 
unitary map between the Hilbert spaces would therefore simply be 
$T_s\mapsto T^+_s\otimes T^-_s$. 

This can be justified abstractly as follows: From all possible area operators  pick a 
maximal commuting subset ${\rm Ar}^\pm_\alpha$ using the axiom of choice (i.e. 
pick a corresponding maximal set of surfaces $S_\alpha$). We may 
construct an adapted orthonormal basis $T^\pm_\lambda$ diagonalising all of 
them\footnote{ If the maximal set still separates the points of the classical 
configurations space, this should leave no room for degeneracies, that is 
the $\lambda_\alpha$ completely specify the eigenvector. We will assume this to be 
the case for the following argument.} 
such that ${\rm Ar}^\pm_\alpha T^\pm_\lambda=\lambda_\alpha T^\pm_\lambda$. 
Now the constraint 
$$
{\rm Ar}^+_\alpha\otimes \mathbb{1}=\mathbb{1}\otimes {\rm Ar}^-_\alpha      
$$
can be solved on vectors $T^+_{\lambda_+}\otimes T^-_{\lambda_-}$ by demanding
$\lambda_+=\lambda_-$. The desired unitary map would then be $T_\lambda\mapsto
T^+_\lambda\otimes T^-_\lambda$. Thus the question boils down to asking whether a maximal
closing subset can be chosen such that the eigenvalues $\lambda$ are just 
the spin networks $s$. We leave this to future research. 
\item In $D\not=3$ the afore mentioned split into selfdual and anti-selfdual sector is meaningless
and we must stick with the dimension independent scheme outlined above. An astonishing
feature of this scheme is that 
after the proposed implementation of the simplicity constraints, the size of the kinematical Hilbert space {\it is the same for all dimensions $D \geq 3$!} By ``size'', we mean that the spin networks are labelled by the same sets of quantum numbers on the graphs.
Of course, before imposing the spatial diffeomorphism constraint
these graphs are embedded into spatial slices of different dimension and thus provide 
different amounts  of degrees of freedom. 
However, after implementation of the diffeomorphism constraint, most of the embedding information will be lost and the graphs can be treated almost as abstract combinatorial objects. 
Let us neglect here, for the sake of the argument, the possibility of certain remaining moduli, depending on the amount
of diffeomorphism invariance that one imposes,  which could a priori be different in different dimensions. In the case that the proposed quantisation would turn out to be correct, that is,
allow for the correct semiclassical limit,
this would mean that {\it the dimensionality of space would be an emergent concept} dictated by the choice of semiclassical states which provide the necessary embedding information.
A possible caveat to this argument is the remaining Hamiltonian constraint and the 
algebra of Dirac observables which critically
depend on the dimension (for instance through the volume operator or dimension dependent 
coefficients, see \cite{BTTI,BTTII}) and which could require to delete 
different amounts of degrees of freedom depending on the dimension.\\
\\
This idea of dimension emergence  is not new in the field of quantum gravity, however,
it is interesting to possibly see here a concrete technical realisation which appears to  
be forced on us by demanding anomaly freedom of the simplicity constraint operators.
Of course, these speculations should be taken with great care: The number of degrees of 
freedom of the classical theory certainly {\it does} strongly  depend on the dimension and therefore 
the speculation of dimension emergence could fail exactly when we try to construct
the semiclassical sector with the solutions to the simplicity constraints advertised above.
This would mean that our scheme is wrong. On the other hand, there are indications 
 \cite{FloriSemiclassicalAnalysisOf} that the semiclassical sector of the LQG Hilbert space already in $D=3$ is entirely described in terms of 6-valent vertices. Therefore, the higher valent graphs 
 which in $D=3$ could correspond to pure quantum degrees of freedom, could account
 for the semiclassical degrees of freedom of higher dimensional General Relativity. Since 
 there is no upper limit to the valence of a graph, this would mean that already the $D=3$ theory
 contains all higher dimensional theories! \\
  \\
Obviously, this puzzle asks for thorough investigation in future research.   
\item The discussion reveals that we should compare the amount of degrees of freedom 
that the classical and the quantum simplicity constraint removes. This is a difficult 
subject, because there is no well defined scheme that attributes quantum to classical 
degrees of freedom unless the Hilbert space takes the form of a tensor product,
where each factor corresponds to precisely one of the classical configuration degrees of freedom. 
The following ``counting'' therefore is highly heuristic and speculative:\\
\\ 
 In the case $D=3$, the classical simplicity constraints remove $6$ degrees of freedom from the constraint surface per point on the spatial slice. In order to count the quantum degrees of freedom
 that are removed by the quantum simplicity constraint when acting on a spin network function, 
 we make the following, admittedly naive analogy:\\
 We attribute to a point on the spatial slice an $N$-valent vertex $v$ of the underlying 
 graph $\gamma$ which is attributed to the spatial slice. This point is equipped with 
 degrees of freedom labelled by edge representations and the intertwiner. Every edge incident at $v$ is shared by exactly one other vertex (or returns to $v$ which however does not change the result). Therefore, only half of the degrees of freedom of an edge can be attributed to 
one vertex.\\ 
We take as edge degrees of freedom the $\lfloor \frac{D+1}{2} \rfloor$ Casimir eigenvalues 
of SO$(D+1)$ labelling the irreducible representation. The edge simplicity constraint removes 
all but one of these Casimir eigenvalues, thus per edge $\lfloor \frac{D-1}{2} \rfloor$ edge 
degrees of freedom are removed. Further, a gauge invariant  intertwiner is labelled by 
a recoupling scheme involving $N-3$ irreducible representations not fixed by the irreducible 
representations carried by the edges adjacent to the vertex in question, which are fully attributed
to the vertex (there are $N-2$ virtual edges coming from coupling 1,2 then 3 etc. until N but the last
one is fixed due to gauge invariance). We take as vertex degrees of freedom these 
$N-3$ irreducible representations each of which is labelled again by $\lfloor \frac{D+1}{2} \rfloor$
Casimir eigenvalues. The vertex simplicity constraint again deletes all but one of these eigenvalues, thus it removes $(N-3) \lfloor \frac{D-1}{2} \rfloor$ quantum degrees of freedom.
We conclude that the quantum simplicity constraint removes 
\be
(N-3+\frac{N}{2}) \lfloor \frac{D-1}{2} \rfloor
\ee
quantum degrees of freedom per point ($N$-valent vertex) where $N-3$ accounts for the vertex and $N/2$ for the $N$ edges counted with half weight  as argued above. This is to be compared with the classical simplicity constraint which removes 
$D^2(D-1)/2-D$ degrees of freedom per point. Requiring equality we see that vertices
of a definitive valence $N_D$ are preferred in $D$ spatial dimensions which for large $D$ grows 
quadratically with D. Specifically for $D=3$ we find $N_3=6$. Thus, our naive 
counting astonishingly yields the same preference for 6-valent graphs in $D=3$ as has been
obtained in \cite{FloriSemiclassicalAnalysisOf} by completely different methods. 
From the analysis of \cite{FloriSemiclassicalAnalysisOf}, it transpires 
that $N_3=6$ has an entirely geometric origin and one thus would rather expect $N_D=2D$
(hypercubulations) and this may indicate that our counting is incorrect.
\end{enumerate}

\section{The Linear Simplicity Constraint Operators}

\label{sec:TheLinearSimplicity}

\subsection{Regularisation and Anomaly Freedom}
\label{sec:RegularisationAndAnomaly}

In \cite{BTTVI}, the connection formulation sketched at the beginning of the previous chapter was altered in that it contains linear simplicity constraints $\epsilon_{IJKL\overline{M}} N^I \pi^{a JK}\approx 0$ and an independent normal $N^I$ as phase space variables. The normal Poisson-commutes with both the connection $A_{aIJ}$ and its momentum $\pi^{aIJ}$ and has its own canonical momentum $P_{I}$. The necessity for this independent normal did not stem from the anomaly encountered when looking at the quadratic quantum simplicity constraints, but from the observation that it was needed to extend the connection formulation to higher dimensional supergravities. 

Since the linear simplicity constraint is a vector of density weight one, it is most naturally smeared over $(D-1)$-dimensional surfaces. The regularisation of the objects
\be
	S^b(S) := \int_S ~ b^{L\overline{M}}(x) \epsilon_{IJKL\overline{M}} N^I(x) \pi^{a JK}(x) \epsilon_{a b_1 ...b_{D-1}} dx^{b_1} \wedge ... \wedge dx^{b_{D-1}} \text{,}
\ee
where $S^b$ denotes the linear simplicity constraint, $S$ a $D-1$-surface, and $b^{L\overline{M}}$ an arbitrary semianalytic smearing function of compact support, therefore is completely analogous to the case of flux vector fields. The corresponding quantum operator 
\be
\hat{S}^b(S)f = \hat{Y}^{\epsilon b \hat N}(S)f = p^*_{\gamma_S} \hat{Y}^{\epsilon b \hat N}_{\gamma_S}(S)f_{\gamma_S}  = p^*_{\gamma_S} \sum_{e \in \gamma_{S}} \epsilon(e,S)  \epsilon_{IJKL\overline{M}} b^{L\overline{M}}(b(e)) \hat N^I(b(e)) R_e^{JK}f_{\gamma_{S}}
\ee 
has to annihilate physical states for all surfaces $S \subset \sigma$ and all semianalytic functions $b^{I\overline{M}}$ of compact support, where $p_{\gamma}$ denotes the cylindrical projection and $\gamma_S$ is a graph adapted to the surface $S$. Since we can always choose surfaces which intersect a given graph only in one point, this implies that the constraint has to vanish when acting on single points of a given graph. In \cite{BTTIII}, it has been shown that the right invariant vector fields actually are in the linear span of the flux vector fields. Therefore, it is necessary and sufficient to demand that
\be
 \epsilon_{IJKL\overline{M}} ~ b^{L\overline{M}}(b(e))~ \hat N^I(b(e)) ~R_e^{JK} \cdot f_{\gamma} = 0 \label{eq:QuantumLinearSimplicity}
\ee
for all points of $\gamma$ (which can be be seen as the beginning point of edges by suitably subdividing and inverting edges). Since $\hat N^I$ acts by multiplication and commutes with the right invariant vector fields, see \cite{BTTVI} for details, the condition is equivalent to\footnote{Use the decomposition of $X_{IJ}$ into its rotational ($\bar{X}_{IJ} := \bar{\eta}_I^K  \bar{\eta}_J^L X_{KL}$) and ``boost" parts (${\bar{X}_I := -\zeta N^J X_{IJ}}$) with respect to $N^I$ in (\ref{eq:QuantumLinearSimplicity}), where $\bar{\eta}^{IJ} = \eta^{IJ} - \zeta N^I N^J$ and $\zeta = 1$ for SO$(D+1)$ and $\zeta = -1$ for SO$(1,D)$ as internal gauge groups. It follows that $\bar{\eta}^{IJ} N_J = 0$.}
\be
	\bar{R}_e^{IJ} \cdot f_{\gamma} = 0\text{,} \label{eq:ActionLinearSimplicity}
\ee
i.e. the generators of rotations stabilising $N^I$ have to annihilate physical states. Before imposing these conditions on the quantum states, we have to consider the possibility of an anomaly. Classically and before using the singular smearing of holonomies and fluxes, both, the linear and the quadratic simplicity constraints are Poisson self-commuting. The quadratic constraint is known to be anomalous both in the Spin Foam \cite{EngleLoopQuantumGravity} as well as in the canonical picture \cite{WielandComplexAshtekarVariables, BTTIII} and thus should not be imposed strongly. 
Also the linear simplicity constraint is anomalous when using a non-zero Immirzi parameter (at least if $\gamma \neq 1$ in the Euclidean theory. But $\gamma = 1$ is ill-defined for SO$(4)$, see e.g., \cite{RovelliLecturesOnLoop}). Surprisingly, in the case at hand and without an Immirzi parameter in four dimensions, we do not find an anomaly.
However that is just because the generators of rotations stabilising $N^I$ form a closed subalgebra! Direct calculation yields, choosing (without loss of generality) $\gamma_{SS'}$ to be a graph adapted to both surfaces $S$, $S'$,
\ba
	\left[\hat{S}^b_{\gamma_{SS'}}(S),\hat{S}^{b'}_{\gamma_{SS'}}(S')\right] f_{\gamma_{SS'}}&=&
	\left[\sum_{e \in \gamma_{SS'}} \hdots \bar{R}_e^{IJ},\sum_{e' \in \gamma_{SS'}}\hdots \bar R_{e'}^{AB} \right] f_{\gamma_{SS'}} = \sum_{e \in \gamma_{SS'}} \hdots \left[ \bar R_e^{IJ}, \bar R_{e}^{AB} \right] f_{\gamma_{SS'}}  \nonumber \\
&=& \sum_{e \in \gamma_{SS'}} \hdots \bar{\eta}^I\m_K \bar{\eta}^J\m_L \bar{\eta}^A\m_C \bar{\eta}^B\m_D ~ f^{KL ~ CD}\m_{~MN} R_e^{MN} f_{\gamma_{SS'}}  \nonumber \\
&=& \sum_{e \in \gamma_{SS'}} \hdots \bar{\eta}^I\m_K \bar{\eta}^J\m_L \bar{\eta}^A\m_C \bar{\eta}^B\m_D \left(\eta^{L][C}\delta^{D]}\m_{[M}\delta_{N]}\m^{[K}\right) R_e^{MN} f_{\gamma_{SS'}}   \nonumber \\ 
&=&\sum_{e \in \gamma_{SS'}} \hdots \bar R_e^{MN} f_{\gamma_{SS'}} \text{,} 
\ea
where the operator in the last line is in the linear span of the vector fields $\hat{S}^b(S)$. The classical constraint algebra is not reproduced exactly (the commutator does not vanish identically), but the algebra of quantum simplicity constraints closes, they are of the first class. Therefore, strong imposition of the quantum constraints does make mathematical sense.

Note that up to now, we did not solve the Gau{\ss} constraint. The quantum constraint algebra of the simplicity and the Gau{\ss} constraint can easily be calculated and reproduces the classical result
\begin{alignat}{3}
	\m& \left[ \hat{S}^b(S), \hat{G}^{AB}[\Lambda_{AB}] \right] p^*_{\gamma_S} f_{\gamma_{S}} \nonumber \\ 
	=& p^*_{\gamma_S} \left[ \sum_{e' \in E(\gamma_{S}), v=b(e')} \epsilon(e',S)  \epsilon^{IJKL\overline{M}} b_{L\overline{M}}(v) \hat N^I(v) R_{e'}^{JK}, \Lambda_{AB}(v) \left( \sum_{e \in E(\gamma_S), v=b(e)} R_e^{AB} + R_N^{AB} \right) \right] f_{\gamma_{S}}  \nonumber \\ 
	=& p^*_{\gamma_S} \Lambda_{AB}(v) \epsilon^{IJKL\overline{M}} b_{L\overline{M}}(v) \sum_{e \in E(\gamma_S), v=b(e)} \epsilon(e,S) \left( \hat N^I(v) \left[ R_e^{JK},R_{e}^{AB}\right] + \eta^{I[A} N^{B]}(v) R_e^{JK} \right) f_{\gamma_{S}}  \nonumber \\ 
	=& p^*_{\gamma_S} \Lambda_{AB}(v) \epsilon^{IJKL\overline{M}} b_{L\overline{M}}(v) \sum_{e \in E(\gamma_S), v=b(e)} \epsilon(e,S) \left( \hat N^I(v) 2 \eta^{KA} R_e^{JB}+ \eta^{I[A} N^{B]}(v) R_e^{JK} \right) f_{\gamma_{S}}  \nonumber \\
	=& \hat{S}^{(- \Lambda \cdot b)}(S) p^*_{\gamma_S} f_{\gamma_{S}}\text{.}
\end{alignat}
where we used $R_N^{AB} := \frac{1}{2} \left( N^{A} \frac{\partial}{\partial N_B} - N^{B} \frac{\partial}{\partial N_A}\right)$. 
It follows that the simplicity constraint operator does not preserve the Gau{\ss} invariant 
subspace (in other words, as in the classical theory, the Gau{\ss} constraint does not generate an ideal in the constraint algebra).  
This implies that the joint kernel of both Gau{\ss} and simplicity constraint must be a 
proper subspace of the Gau{\ss} invariant subspace. It is therefore most convenient 
to look for the joint kernel on  in the kinematical (non Gau{\ss} invariant) Hilbert space.

\subsection{Solution on the Vertices}
\label{sec:SolutionOnThe}
Consider a slight modification of the usual gauge-variant spin network functions, where the intertwiners $i_v = i_v(N)$ are square integrable functions of $N^I$. Let $v$ be a vertex of $\gamma$ and $e_1,\hdots,e_n$ the edges of $\gamma$ incident at $v$, where all orientations are chosen such that the edges are all outgoing at $v$. Then we can write the modified spin network functions 
\ba
T_{\gamma, \vec{l}, \vec{i}}(A,N) &:=& (i_v(N))_{\vec{K}_1\hdots \vec{K}_n} \prod_{i=1}^n \left(\pi_{l_{e_i}}(h_{e_i}(A))\right)_{\vec{K}_i \vec{K}'_i} (M_v)_{\vec{K}'_1\hdots \vec{K}'_n}  \nonumber \\
&=& \text{tr}\left(i_v(N) \cdot \otimes_{i=1}^n \pi_{l_{e_i}}(h_{e_i}(A))\cdot M_v\right)\text{,}
\ea
where $M_v$ contracts the indices corresponding to the endpoints of the edges $e_i$ and represents the rest of the graph $\gamma$. These states span the combined Hilbert space for the 
normal field and the connection $\mathcal{H}_{T} = \mathcal{H}_{\text{grav}} \otimes \mathcal{H}_N$ (cf. \cite{BTTVI}) and they will prove convenient for solving the simplicity constraints. Choose the surface $S'$ such that it intersects a given graph $\gamma'$ only in the vertex $v' \in \gamma'$. The action of $\hat{S}^{b}(S')$ on the vertex $v'$ of a spin network $T_{\gamma', \vec l, \vec i}(A,N)$ implies with (\ref{eq:ActionLinearSimplicity}) that
\ba
\hat{S}^{b}(S')_{\gamma'} T_{\gamma', \vec{l}, \vec{i}}(A,N) &=& 0 \nonumber \\
\Longleftrightarrow ~ \text{tr} \left( \left( i_v(N) \bar{\tau}_{\pi_{l_e}}^{IJ} \right) \cdot \otimes_{i=1}^n \pi_{l_{e_i}} \left(h_{e_i}(A) \right)\cdot M_v \right) &=& 0 ~\text{$\forall e$ at $v'$,}
\ea
where $\tau^{IJ}_{\pi_{l_e}}$ here denote the generators of SO$(D+1)$ in the representation $\pi_{l_e}$ of the edge $e$ and the bar again denotes the restriction to rotational components (w.r.t. $N^I$). The above equation implies that the intertwiner $i_v$, seen as a vector transforming in the representation $\bar \pi_{l_e}$ dual to $\pi_{l_e}$ of the edge $e$, has to be invariant under the SO$(D)_N$ subgroup which stabilises the $N^I$. By definition \cite{VilenkinSpecialFunctionsAnd}, the only representations of SO$(D+1)$ which have in their space nonzero vectors which are invariant under a SO$(D)$ subgroup are of the representations of class one (cf. also appendix \ref{app:SO(D+1)}), and they exactly coincide with the simple representations used in Spin Foams \cite{FreidelBFDescriptionOf}. It is easy to see that the dual representations (in the sense of group theory) of simple representations are simple representations. Therefore, all edges must be labelled by simple representations of SO$(D+1)$. Moreover, SO$(D)$ is a massive subgroup of SO$(D+1)$ \cite{VilenkinSpecialFunctionsAnd}, so that the (unit) invariant vector $\xi_{l_e}(N)$ in the representation $\bar \pi_{l_e}$ is unique, which implies that the allowed intertwiners $i_v(N)$ are given by the tensor product of the invariant vectors of all $n$ edges and potentially an additional square integrable function $F_v(N)$, $ i_v(N) = \xi_{l_{e_1}}(N) \otimes ... \otimes \xi_{l_{e_{n}}}(N) \otimes F_v(N)$. Going over to normalised gauge invariant spin network functions implies that $F_v (N) = 1$, and the resulting intertwiner space solving the simplicity and Gau{\ss} constraint becomes one-dimensional, spanned by $I_v(N) := \xi_{l_{e_1}}(N) \otimes ... \otimes \xi_{l_{e_{n}}}(N)$. We will call these intertwiners and vertices coloured by them linear-simple. For an instructive example of the linear-simple intertwiners, consider the defining representation (which is simple since the highest weight vector is $\Lambda = (1,0,\hdots,0)$, cf. appendix \ref{app:SO(D+1)}). The unit vector invariant under rotations (w.r.t. $N^I$) is given by $N^I$ and for edges in the defining representation incoming at $v$ we simply contract $h_e^{IJ}N_J$. If the constraint is acting on an interior point of an analytic edge, this point can be considered as a trivial two-valent vertex and the above result applies. Since this has to be true for all surfaces, a spin network function solving the constraint would need to have linear-simple intertwiners at every point of its graph $\gamma$, i.e. at infinitely many points, which is in conflict with the definition of cylindrical functions (cf. \cite{ThiemannKinematicalHilbertSpaces}). In the next section, we comment on a possibility of how to implement this idea.

\subsection{Edge Constraints}

As noted above, the imposition of the linear simplicity constraint operators acting on edges is problematic, because it does not, as one might have expected, single out simple representations, but demand that at every point where it acts, there should be a linear-simple intertwiner. The problem with this type of solution is that all intertwiners, even trivial intertwiners at all interior points of edges, have to be linear-simple, which is however in conflict with the definition of a cylindrical function, in other words, there would be no holonomies left in a spin network because every point would be a $N$-dependent vertex.

It could be possible to resolve this issue using a rigging map construction \cite{HiguchiQuantumLinearizationInstabilities, GiuliniOnTheGenerality, GiuliniAUniqenessTheorem} of the type
\be
\eta(T_{\gamma, \vec{l}, \vec{l}_N, \vec{i}})[T_{\gamma', \vec{l}', \vec{l}'_N, \vec{i}'}] := \lim_{\mathfrak{P}_{\gamma} \ni p_{\gamma} \rightarrow \infty} C\left(p_\gamma, T_\gamma, T_{\gamma'} \right) \left\langle T^{p_\gamma}_{\gamma, \vec{l}, \vec{l}_N, \vec{i}}~,~ T_{\gamma', \vec{l}', \vec{l}'_N, \vec{i}'} \right\rangle_{\text{kin}} \text{,}
\label{eq:Rigging}
\ee
where $\mathfrak{P}_{\gamma}$ is the set of finite point sets $p$ of a graph $\gamma$, $p = \{\{x_i\}_{i=1}^N | x_i \in \gamma \, \forall \, i, N < \infty\}$.  $\mathfrak{P}_{\gamma}$ is partially ordered by inclusion, $q \succeq p$ if $p$ is a subset of $q$, so that the limit is meant in the sense of net convergence with respect to $\mathfrak{P}_\gamma$. By the prescription $T^{p_\gamma}_{\gamma, \vec{l}, \vec{l}_N, \vec{i}}$ we mean the projection of $T_{\gamma, \vec{l}, \vec{l}_N, \vec{i}}$ onto linear-simple intertwiners at every point in $p$ and $ C\left(p_\gamma, T_\gamma, T_{\gamma'} \right)$ is a numerical factor. Assuming this to work, consider any surface $S$ intersecting $\gamma'$. We (heuristically) find 
\ba
\eta(T_\gamma)[\hat{S}^b(S) T_{\gamma'}] &=& \lim_{\mathfrak{P}_{\gamma} \ni p_{\gamma} \rightarrow \infty}  C\left(p_\gamma, T_\gamma, T_{\gamma'} \right) \left\langle T_\gamma^p, \hat{S}^b(S) T_{\gamma'} \right\rangle_{\text{kin}}   \nonumber \\
 &=& \lim_{\mathfrak{P}_{\gamma} \ni p_{\gamma} \rightarrow \infty}  C\left(p_\gamma, T_\gamma, T_{\gamma'} \right) \left\langle [\hat{S}^b(S)]^{\dagger} T_\gamma^p, T_{\gamma'} \right\rangle_{\text{kin}} \nonumber \\
 &=& \lim_{\mathfrak{P}_{\gamma} \ni p_{\gamma} \rightarrow \infty} C\left(p_\gamma, T_\gamma, T_{\gamma'} \right) \left\langle \hat{S}^b(S) T_\gamma^p, T_{\gamma'} \right\rangle_{\text{kin}}  = 0 \text{,}
\ea
since the intersection points of $S$ with $\gamma$ will eventually be in $p_{\gamma}$ and $\hat{S}^b(S)$ is self-adjoint.

We were however not able to find such a rigging map with satisfactory properties. It is especially difficult to handle observables with respect to the linear simplicity constraint and to implement the requirement, that the rigging map has to commute with observables. It therefore seems plausible to look for non-standard quantisation schemes for the linear simplicity constraint operators, at least when acting on edges. Comparison with the quadratic simplicity constraint suggests that also the linear constraint should enforce simple representations on the edges, see the following remarks as well as section \ref{sec:mixed} for ideas on how to reach this goal.

\subsection{Remarks}

\label{sec:LinearRemarks}

The intertwiner space at each vertex is one-dimensional and thus the strong solution of the unaltered linear simplicity constraint operator contrasts the quantisation of the classically imposed simplicity constraint at first sight. A few remarks are appropriate:

\begin{enumerate}
\item One could argue that the intertwiner space at a vertex $v$ is infinite-dimensional by taking into account holonomies along edges $e'$ originating at $v$ and ending in a $1$-valent vertex $v'$. Since $e'$ and $v'$ are assigned in a unique fashion to $v$ if the valence of $v$ is at least $2$, we can consider the set $\{v, e', v'\}$ as a new ``non-local'' intertwiner. Since we can label $e'$ with an arbitrary simple representation, we get an infinite set of intertwiners which are orthogonal in the above scalar product. This interpretation however does not mimic the classical imposition of the simplicity constraints or the above imposition of the quadratic simplicity constraint operators. 

\item The main difference between the formulation of the theory with quadratic and linear simplicity
constraint respectively is the appearance of the additional normal field sector in the linear
case. Thus one could expect that one would recover the quadratic simplicity constraint formulation
by ad hoc averaging the solutions of the linear constraint over the normal field dependence with 
the probability measure $\nu_N$ defined in \cite{BTTVI}. Indeed, if one does so, then one recovers 
 the solutions to the quadratic simplicity constraints in terms of the Barrett-Crane intertwiners
 in $D=3$ and higher dimensional analogs thereof 
 as has been shown long ago by Freidel, Krasnov, and Puzio \cite{FreidelBFDescriptionOf}.
 A similar observation has been made in \cite{BaratinQuantumSimplicialGeometry}.
 Such an average also deletes the solutions with ``open ends'' of the previous remark by an appeal
 to Schur's lemma.
 Since after such an average the $N$ dependence of all solutions disappears, we can drop 
 the $\mu_N$ integral in the kinematical inner product since $\mu_N$ is a probability measure.
 The resulting effective physical scalar product would then be the Ashtekar-Lewandowski scalar product of the theory between the solutions to the quadratic simplicity constraints. 
 Such an averaging would also help with the solution of the edge constraints, since a $2$-valent linear-simple intertwiner is averaged as
\be
\int_{S^D} d \nu(N) \, \bar{\xi}^\alpha_{l}(N) \xi^\beta_{l}(N) = \frac{1}{d_{\pi_l}} \delta^{\alpha \beta} \text{,}
\ee
thus yielding a projector on simple representations.
   
\item It can be easily checked that the volume operator as defined in \cite{BTTIII}, and therefore also more general operators like the Hamiltonian constraint, do not leave the solution space to the linear (vertex) simplicity constraints invariant. A possible cure would be to introduce a projector $\mathcal{P}_S$ on the solution space and redefine the volume operator as $\hat{\mathcal{V}} := \mathcal{P}_S \hat{V} \mathcal{P}_S$. Such procedures are however questionable on the general ground that anomalies can always be removed by projectors. 

\item If one accepts the usage of the projector $\mathcal{P}_S$, calculations involving the volume operator simplify tremendously since the intertwiner space is one-dimensional. We will give a few examples which can be calculated by hand in a few lines, restricting ourselves to the defining representation of SO$(D+1)$, where the SO$(D)_N$ invariant unit vector is given by $N^I$. 

Having direct access to $N^I$, one can base the quantisation of the volume operator on the classical expression
\be
	\det{q} = \left| \frac{1}{D!} \epsilon_{IJ_1...J_D} N^I \left(\pi^{a_1 K_1J_1} N_{K_1}\right)\hdots \left(\pi^{a_D K_DJ_D} N_{K_D}\right) \epsilon_{a_1...a_D}  \right|^{\frac{1}{D-1}}\text{.}
\ee
In the case $D+1$ uneven, this choice is much easier than the expression quantised in \cite{BTTIII}. In the case $D+1$ even, the above choice is of the same complexity\footnote{Up to $(N)^{D+1}$, but in the chosen representation $\hat N$ acts by multiplication and therefore is less problematic than additional powers of right invariant vector fields.} as the one in \cite{BTTIII}, but leads to a formula applicable in any dimension and therefore, for us, is favoured. Proceeding as in \cite{BTTIII}, we obtain for the volume operator
\ba
 \hat{V}(R) &=& \int_R d^Dp \, \widehat{\left| \det (q)(p) \right| \m_\gamma }= \int_R d^Dp \, \hat{V}(p)_\gamma \text{,} \\
 \hat{V}(p) &=& \left( \frac{\hbar}{2} \right)^{\frac{D}{D-1}} \sum_{v \in V(\gamma)} \delta^D(p,v) \hat{V}_{v,\gamma} \text{,} \\
 \hat{V}_{v,\gamma} &=& \left| \frac{i^D}{D!} \sum_{e_1, \ldots, e_D \in E(\gamma), \, e_1 \cap \ldots \cap e_D = v} s(e_1, \ldots, e_D) \hat{q}_{e_q, \ldots, e_D} \right|^{\frac{1}{(D-1)}} \text{,} \\
 \hat{q}_{e_1, \ldots, e_D} &=& \epsilon_{IJ_1\hdots J_D} \hat N^I \left(R_{e_1}^{K_1J_1} \hat N_{K_1}\right)\hdots \left(R_{e_D}^{K_DJ_D} \hat N_{K_D}\right) \text{.}
\ea

Note that the operator $\hat q_{e_1,\hdots,e_D}$ is built from $D$ right invariant vector fields. Since these are antisymmetric, $\hat q_{e_1,\hdots,e_D}^T = (-1)^D \hat q_{e_1,\hdots,e_D}$. In the case at hand, we have to use the projectors $\mathcal{P}_S$ to project on the allowed one-dimensional intertwiner space, the operator $\mathcal{P}_S \hat q \mathcal{P}_S$ therefore has to vanish for the case $D + 1$ even (an antisymmetric matrix on a one-dimensional space is equal to 0). However, the volume operator depends on $\hat{q}^2$, and $\mathcal{P}_S \hat q ^2 \mathcal{P}_S$ actually is a non-zero operator in any dimension, though trivially diagonal. Therefore, also $\hat{\mathcal{V}}$ is diagonal. 

The simplest non-trivial calculation involves a $D$-valent non-degenerate (i.e. no three tangents to edges at $v$ lie in the same plane) vertex $v$ where all edges are labelled by the defining representation of SO$(D+1)$ and thus the unique intertwiner which we will denote by $\left| N^{A_1} \hdots N^{A_D} \right \rangle$. We find 
\ba
 \hat{q}_{e_1,\hdots,e_D} \left| N^{A_1} \hdots N^{A_D} \right \rangle &=& s{(e_1,...,e_D)} \left(-\frac{1}{2} \right)^{D} \left| N_I \epsilon^{IA_1...A_D} \right \rangle \nonumber \text{,} \\
 \hat{q}_{e_1,\hdots,e_D} \left| N_I \epsilon^{I A_1 \ldots A_D} \right \rangle &=& s{(e_1,...,e_D)}  \left(\frac{1}{2} \right)^{D} D! \left| N^{A_1} \hdots N^{A_D} \right \rangle \nonumber \text{,}  \\
 \hat{V}_v \left| N^{A_1} \hdots N^{A_D} \right \rangle &=& \left( \left( \frac{1}{4} \right)^{D} D! \right)^{\frac{1}{2(D-1)}} \left| N^{A_1} \hdots N^{A_D} \right \rangle \text{,}
\ea
i.e. for those special vertices, the volume operator preserves the simple vertices. For vertices of higher valence and/or other representations, we need to use the projectors. Of special interest are the vertices of valence $D+1$ (triangulation) and $2D$, where every edge has exactly one partner which is its analytic continuation through $v$ (cubulation). We find
\ba
\hat{\mathcal{V}}_v \left| N^{A_1} \hdots N^{A_{D+1}} \right \rangle &=& \left( \left( \frac{1}{4} \right)^{D} (D+1)! \right)^{\frac{1}{2(D-1)}} \left| N^{A_1} \hdots N^{A_{D+1}} \right \rangle \nonumber \text{,} \\
\hat{\mathcal{V}}_v \left| N^{A_1} \hdots N^{A_{2D}}, \text{cubic} \right \rangle &=& \left( \left( \frac{1}{2} \right)^{D} (D)! \right)^{\frac{1}{2(D-1)}} \left| N^{A_1} \hdots N^{A_{2D}}, \text{cubic} \right \rangle \text{.} 
\ea
The dimensionality of the spatial slice now appears as a quantum number like the spins labelling the representations on the edges and it could be interesting to consider a large dimension limit in the spirit of the large $N$ limit in QCD.

\item \label{rem:LinearAnomaly} When introducing an Immirzi parameter in $D=3$ \cite{BTTII}, i.e. using the linear constraint $\epsilon_{IJKL} N^J \pi^{aKL} \approx 0$ while having $\{ A_{aIJ}(x), \mbox{}^{(\gamma)}\pi^{bKL}(y)\} = 2 \delta_a^b  \delta_{[I}^K \delta_{J]}^L \delta^{(D)}(x-y)$ with $ \mbox{}^{(\gamma)}\pi^{aIJ} = \pi^{aIJ} + 1/(2 \gamma) \epsilon^{IJKL} \pi^{a} \mbox{}_{KL}$, the linear simplicity constraint operators become anomalous unless $\gamma = \pm \sqrt{\zeta}$, the (anti)self-dual case, which however results in non-invertibility of the prescription $\mbox{}^{(\gamma)}$. Repeating the steps in section \ref{sec:RegularisationAndAnomaly}, we find that these anomalous constraints require $\epsilon_{IJKL} N^I (R_e^{KL} - 1/(2\gamma) \epsilon^{KLMN} R_e^{MN}) \cdot f_{\gamma} = 0$. Since $\epsilon_{IJKL} N^I (R_e^{KL} - 1/(2\gamma) \epsilon^{KLMN} R_e^{MN})$ do not generate a subgroup, the constraint can not be satisfied strongly if the edge $e$ transforms in an irreducible representation of SO$(D+1)$ (by definition, the representation space does not contain an invariant vector).

\end{enumerate}

In order to figure out the ``correct'' quantisation, one can try, in analogy to the strategy for the quadratic simplicity constraints, to weaken the imposition of the constraints at the quantum level. The basic difference between the linear and the quadratic simplicity constraints is that the time normal $N^I$ is left arbitrary in the quadratic case and fixed in the linear case. In order to loose this dependence in the linear case, one could average over all $N^I$ at each point in $\sigma$, which however leads to the Barrett-Crane intertwiners as described above. In analogy to the quadratic constraints,  we could choose the subset 
\begin{eqnarray}
    \epsilon_{IJKL \overline{M}} N^J  \left(R_{e_1}^{KL} + R_{e_2}^{KL}  \right)= 0 \nonumber \\
  \epsilon_{IJKL \overline{M}}N^J  \left(R_{e_1}^{KL} + R_{e_2}^{KL} +R_{e_3}^{KL}  \right)= 0 \nonumber \\
  \ldots \nonumber \\
   \epsilon_{IJKL \overline{M}} N^J  \left(R_{e_1}^{KL} +\ldots+R_{e_{N-2}}^{KL}  \right)= 0 \label{eq:offdiagonallinearsimplicities}
\end{eqnarray} 
for each $N$-valent vertex plus the edge constraints. As above, the choice of the subset specifies a recoupling scheme and the imposition of the constraints leads to the contraction of the virtual edges and virtual intertwiners of the recoupling scheme with the SO$(D)_N$-invariant vectors $\xi_{l_{e_i}}(N)$ and their complex conjugates $\bar{\xi}_{l_{e_i}}(N)$, see fig. \ref{fig:linearsubsetimposition}. Gauge invariance can still be used at each (virtual) vertex in this calculation in the form $\sum_{i} \bar{R}_{e_i}=0$, which is sufficient since only $\bar{R}_{e_i}$ appears in the linear simplicity constraints. If we now integrate over each pair of $\xi_{l_{e_i}}(N)$ ``generated'' by the elements of the proposed subset of the simplicity constraint operators separately, we obtain projectors on simple representations for each of the virtual edges in the recoupling scheme. The integration over $N^I$ for the edge constraints yields projectors on simple representations in the same manner. Finally, we obtain the simple intertwiners of the quadratic operators in addition to solutions where incoming edges are contracted with SO$(D)_N$-invariant vectors $\xi_{l_{e_i}}(N)$. A few remarks are appropriate:

\begin{figure}[h]
\centering
\includegraphics[trim = 0mm 0mm 0mm 170mm, clip, scale=0.5]{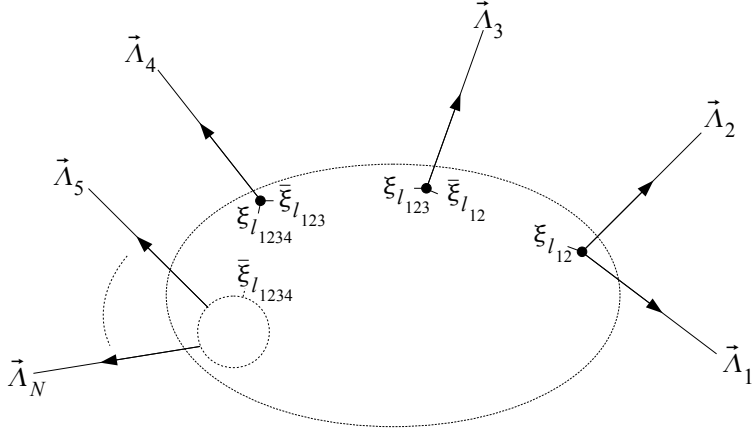}
\caption{Recoupling scheme corresponding to the subset of linear vertex simplicity constraint operators (\ref{eq:offdiagonallinearsimplicities}).}
\label{fig:linearsubsetimposition}
\end{figure}

\begin{enumerate}
\setcounter{enumi}{5}
\item Although this procedure yields a promising result, it contains several non-standard and ah-hoc steps which have to be justified. One could argue that the ``correct'' quantisation of the linear and quadratic simplicity constraints should give the same quantum theory, however, as is well known, classically equivalent theories result in general in non-equivalent quantum theories, which nevertheless can have the same classical limit.

\item It is unclear how to proceed with ``integrating out'' $N^I$ in the general case. For the vacuum theory, integration over every point in $\sigma$ gives the Barrett-Crane intertwiner for the edges contracted with SO$(D)_N$-invariant vectors. This type of integration would also get rid of the $1$-valent vertices and thus allow for a natural unitary map to the quadratic solutions as already
mentioned above. 

\item When introducing fermions, there is the possibility for non-trivial gauge-invariant functions of $N^I$ at the vertices which immediately results in the question of how to integrate out this $N^I$-dependence. Next to including those $N^I$ in the above integration or to integrate out the remaining $N^I$ separately, one could transfer this integration back into the scalar product. Since the authors are presently not aware of an obvious way to decide about these issues, we will leave them for further research. 

\end{enumerate}

\subsection{Mixed Quantisation}
\label{sec:mixed}

Since the implementation of the quadratic simplicity constraints described above yields a more promising result than the implementation of the linear constraints, we can try to perform a mixed quantisation by noting that we can classically express the linear constraints for even $D$ in the form 
\be
 \frac{1}{4} \epsilon_{IJKL\overline{M}} \pi^{aIJ} \pi^{bKL} \approx 0, ~~~ N^I-n^I(\pi) \approx 0 \text{.}
\ee
The phase space extension derived in \cite{BTTVI} remains valid when interchanging the linear simplicity constraint for the above constraints. The reason for restricting $D$ to be even is that we have an explicit expression for $n^I(\pi)$, see \cite{BTTI,BTTII}. Since a quantisation of $n^I(\pi)$ will most likely not commute with the Hamiltonian constraint operator, we resort to a master constraint. Note that the expression
\be
M'_N := \frac{\left((N^I-n^I(\pi)) \sqrt{q}^{D-1} \right) \delta_{IJ} \left((N^J-n^J(\pi)) \sqrt{q}^{D-1}\right)}{\sqrt{q}^{2D-3}} \text{,}
\ee
which is the densitised square of $N^I-n^I(\pi)$, can be quantised as 
\be
 \hat{M'}_N = 2 \widehat{\sqrt{q}^{3-2D}} ( \sqrt{| \hat{V}^I \hat{V}_I |} - N_I \hat{V}^I) \text{,}
\ee
when using a suitable factor ordering, where a quantisation of $\sqrt{q}^{3-2D}$ is described in \cite{BTTIII}. The solution space
is not empty since the intertwiner
\be
 s{(e_1,...,e_D)} \sqrt{D!} |N^{A_1}N^{A_2} \ldots N^{A_D} > + | N_B \epsilon^{B A_1 \ldots A_D}>
\ee
is annihilated by $\hat{M}_N$, which can be easily checked when using the results of the volume operator acting on the solution space of the full set of linear simplicity constraint operators. In order to turn the expression into a well defined master constraint operator, we have to square it again and to adjust the density weight, leading to 
\be
 \hat{\text{\bf M}}_N = 4 \left(  \widehat{\sqrt{q}^{5/2-2D}} ( \sqrt{| \hat{V}^I \hat{V}_I |} - N_I \hat{V}^I)  \right)^\dagger \widehat{\sqrt{q}^{5/2-2D}} ( \sqrt{| \hat{V}^I \hat{V}_I |} - N_I \hat{V}^I) \text{,}
\ee
which is by construction a self-adjoint operator with non-negative spectrum. We remark that it was necessary to use the fourth power of the classical constraint for quantisation, because the second power, having the desired property that its solution space is not empty, does not qualify as a well defined master constraint operator in the ordering we have chosen. There exists however no a priori reason why one should not take into account master constraint operators constructed from higher powers of classical constraints \cite{ThiemannQSD8}. 
Curiously,  the quadratic simplicity constraint operators as given above do not annihilate the solution 
displayed.
Clearly, the calculations will become much harder as soon as vertices with a valence higher than $D$ are used, since the building blocks of the volume operator will not be diagonal on the intertwiner space.  

This type of quantisation is further discussed in section \ref{subsec:EPRL}, where a possible application to using EPRL intertwiners is outlined. In contrast to the earlier assumption of $D$ being even in order to have an explicit expression $n^I(\pi)$, we can also perform the mixed quantisation using $ n^I n_J \approx \frac{1}{D-1} \left( \pi^{aKI} \pi_{aKJ}-\zeta \eta^{I} \m_J \right)$ for Euclidean internal signature $\zeta = 1$ and the constraint $N_J(n^I n^J (\pi) - N^I N^J) \approx 0$. For the application proposed, we will only need that the corresponding master constraint can be regularised such that it vanishes when not acting on non-trivial vertices, which can be achieved as before.

\section{Comparison to Existing Approaches}

\label{sec:Comparison}

In this section, we are going to comment on the relation of existing results from the spin foam literature to the proposals in this paper. In short, the main conclusion will be that in the case of four spacetime dimensions, many results from the spin foam literature can be used also in the canonical framework. However, they fail to work in higher dimensions due to special properties of the four dimensional rotation group which are heavily used in spin foams. We will not comment on results based on coherent state techniques \cite{FreidelANewSpin, DupuisRevisitingTheSimplicity, DupuisHolomorphicSimplicityConstraints, DupuisHolomorphicLorentzianSimplicity} since we do not see a resemblance to our results which do not make use of coherent states in any way. Nevertheless, similarities could be present as the relation between the EPRL \cite{EngleLoopQuantumGravity} and FK \cite{FreidelANewSpin} models show.

\subsection{Continuum vs. Discrete Starting Point}

The starting point for introducing the simplicity constraints in the spin foam models is the reformulation of general relativity as a BF theory subject to the simplicity constraints, and thus similar to the point of view taken in this series of papers. The crucial difference however is that while spin foam models start classically from discretised general relativity, the canonical approach discussed here starts from its continuum formulation. When looking at the simplicity constraints, this difference manifests itself in the choice of $(D-1)$-surfaces over which the the generalised vielbeins (i.e. the bivectors in spin foam models) have to be smeared. Starting from a discretisation of spacetime, the set of $(D-1)$-surfaces is fixed by foliating the discretised spacetime. Restricting to a simplicial decomposition of a four-dimensional spacetime as an example, these would e.g. be the faces $f$ of a tetrahedron $t$ in the boundary of the discretisation. It follows that one can take the bivectors $B^{IJ}$ integrated over the individual faces of a tetrahedron, $B^{IJ}_f(t) := \int_f  B^{IJ}$, as the basic variables and the quadratic (off)-diagonal simplicity constraints read \cite{EngleLoopQuantumGravity}
\ba
	C_{ff} := \epsilon_{IJKL} B_f(t)^{IJ}  B_f(t)^{KL} ~~~ & & \text{diagonal simplicity}, ~ \forall f \in t \\
	C_{f f'} := \epsilon_{IJKL} B_f(t)^{IJ}  B_{f'}(t)^{KL} ~~~ & & \text{off-diagonal simplicity}, ~ \forall f, f' \in t \label{eq:Off-diagonalSimplicitySpinfoam} \text{.}
\ea
In the continuum formulation however, we have to consider all possible $(D-1)$-surfaces, and thus also hypersurfaces containing the vertex $v$ dual to the tetrahedron $t$. The resulting flux operators a priori contain a sum of right invariant vector fields $R_e^{IJ}$ acting on all the edges $e$ connected to $v$. While this poses no problem for the diagonal simplicity constraints which act on edges of the spin networks as shown in \cite{BTTIII}, the off-diagonal simplicity constraints arising when both surfaces contain $v$ are not given by (\ref{eq:Off-diagonalSimplicitySpinfoam}), but by sums over different $C_{f f'}$, see \cite{BTTIII} for details. It can however be shown by suitable superpositions of simplicity constraints associated to different surfaces that (\ref{eq:Off-diagonalSimplicitySpinfoam}) is actually implied also by the quadratic simplicity constraints arising from a proper regularisation in the canonical framework. This statement is non-trivial and had to be proved in \cite{BTTIII}. Thus, we can also in the canonical theory consider the individual building blocks (\ref{eq:allsimplicities}) as done in section \ref{sec:QuadraticSimplicity} of this paper. Furthermore, the same is also true when using linear simplicity constraints, i.e. the properly regularised linear simplicity constraints in the canonical theory imply that all building blocks (\ref{eq:QuantumLinearSimplicity}) vanish. 

We also note that there is no analogue of the normalisation simplicity constraints \cite{FreidelBFDescriptionOf} in the canonical treatment since the generalised vielbeins do not have timelike tensorial indices after being pulled back to the spatial hypersurfaces.

\subsection{Projected Spin Networks}

Projected spin networks were originally introduced in \cite{AlexandrovHilbertSpaceStructure, LivineProjectedSpinNetworks} to describe Lorentz covariant formulations of Loop Quantum Gravity, meaning that the internal gauge group is SO$(1,3)$ (or SL$(2, \mathbb{C})$) instead of SU$(2)$. The basic idea is that next to the connection, the time normal field, often called $x$ or $\chi$ in the Spin Foam literature, becomes a multiplication operator since it Poisson-commutes classically with the connection. Since the physical degrees of freedom of Loop Quantum Gravity formulated in terms of the usual SU$(2)$ connection and its conjugate momentum are orthogonal to the time normal field, one performs projections in the spin networks from the full gauge group SO$(1,3)$ to a subgroup stabilising the time normal. Since the projector transforms covariantly under SO$(1,3)$, a (gauge invariant) projected spin network is already defined by its evaluation for a specific choice of the time normals and the resulting effective gauge invariance is only SU$(2)$, which exemplifies the relation to the usual SU$(2)$ formulation in the time gauge $x^I (=N^I) = (1,0,0,0)$. 

Despite its close relation to the techniques used in this paper and its merits for the four-dimensional treatment, there are several problems connected with using this approach in the canonical framework discussed in this series of papers which we will explain now. While the extension of projected spin networks to different gauge groups has already been discussed in \cite{LivineProjectedSpinNetworks}, there is a subtle problem associated with the part of the connection which is projected out by the projections, that could not have been anticipated by looking at Loop Quantum Gravity in terms of the Ashtekar-Barbero variables. There, the physical information in the connection, the extrinsic curvature, is located in the rotational components of the connection. To see this, consider in four dimensions the 2-parameter family of connections discussed in \cite{BTTII}\footnote{Note that the definitions of the parameters are different in \cite{BTTII} for calculational simplicity, but here we prefer this parametrisation to make our point clear.}, 
\be
	A_{aIJ} = \Gamma^{\text{hyb}}_{aIJ} + \beta K_{aIJ} + \gamma \frac{1}{2} \epsilon_{IJ} \m^{KL} K_{aKL} \text{,}
\ee
where $\gamma$ corresponds to the Barbero-Immirzi parameter restricted to four dimensions and $\beta$ is the new free parameter appearing in any dimension. $K_{aIJ}$ decomposes as \cite{BTTI, BTTII}
\be
	K_{aIJ} = 2 N_{[I} \bar{K}_{a|J]} + \bar{K}^{\text{trace}}_{aIJ} + \bar{K}^{\text{trace free}}_{aIJ} \text{,}
\ee 
where $\bar{K}_{aIJ}$ means that $N^I \bar{K}_{aIJ}  = N^J \bar{K}_{aIJ} =0$ and the trace / traceless split is performed with respect to the hybrid vielbein. The extrinsic curvature which we need to recover from $A_{aIJ}$ is located in $\bar{K}_{aJ}$, whereas $\bar{K}^{\text{trace}}_{aIJ}$ vanishes by the Gau{\ss} constraint and $\bar{K}^{\text{trace free}}_{aIJ}$ is pure gauge from the simplicity gauge transformations. 

Now setting $\beta=0$ and $N^I = (1,0,0,0)$ in four dimensions, we recover the Ashtekar-Barbero connection and see that the physical information is located in the rotational components of $A_{aIJ}$. It thus makes sense to project onto this subspace in the projected spin network construction, i.e. we are not loosing physical information. On the other hand, setting $\gamma = 0$ in four dimensions or going to higher dimensions, we see that a projection onto the subspace orthogonal to $N^I$ annihilates the physical components of the connection. This would not be necessarily an issue if one would just project the projected spin network at the intertwiners, but when one tries to go to fully projected spin networks as proposed in \cite{AlexandrovHilbertSpaceStructure}. Then, since one would take a limit of inserting projectors at every point of the spin network, the physical information in the connection would be completely lost. 

Next to this problem, there are other problems associated to taking an infinite refinement limit for projected spin networks as discussed by Alexandrov \cite{AlexandrovHilbertSpaceStructure} and Livine \cite{LivineProjectedSpinNetworks}, e.g. that fully projected spin networks are not spin networks any more (since they only contain vertices and no edges) and, connected with this problem, that the trivial bivalent vertex, the Kronecker delta, is not an allowed intertwiner. Similar problems have been encountered in section \ref{sec:TheLinearSimplicity}, i.e. while the vertex simplicity constraints could be solved by a construction very similar to projected spin networks where one projects the incoming and outgoing edges at the intertwiner {\it in the direction} of the time normal $N^I$, imposing the linear simplicity constraint on the edges, one would have to insert ``trivial'' bivalent vertices of the form $N^I N^J$ at every point of the spin network, whereas one would need to insert the the Kronecker delta $\delta^{IJ}$ to achieve cylindrical consistency while maintaining a spin network containing edges and not only vertices. 

Thus, the main problem with using (fully) projected spin networks is connected to the fact that we do not know of an analogue of the Barbero-Immirzi parameter in higher dimensions which would allow us to put the extrinsic curvature also in the rotational components of the connection. In four dimensions on the other hand, this problem would be absent and one would be left with the issue of refining the projected spin networks, which is however also present in section \ref{sec:TheLinearSimplicity} of this paper. Therefore, using projected spin networks in four dimensions with non-vanishing Barbero-Immirzi parameter is an option for the canonical framework developed in this series of papers and the known issues discussed above should be addressed in further research.

\subsection{EPRL Model}

\label{subsec:EPRL}

The basic idea of the EPRL model is to implement the diagonal simplicity constraints as usual, but to replace the off-diagonal simplicity constraints by linear simplicity constraints which are implemented with a master constraint construction \cite{EngleLoopQuantumGravity} or weakly \cite{DingTheVolumeOperator}. Furthermore, the Barbero-Immirzi parameter is a necessary ingredient. We restrict here to the Euclidean model since its group theory is much closer to the connection formulation with compact gauge group SO$(D+1)$. While the diagonal simplicity constraints give the well known relation
\be
	(j^+)^2 = \left( \frac{\gamma+1}{\gamma-1} \right)^2 (j^-)^2 \text{,}
\ee
the master constraint for the linear constraints gives \cite{EngleLoopQuantumGravity}, up to $\hbar$ corrections\footnote{Note that these $\hbar$ corrections are necessary since the master constraint, by construction, has the same solution space as the original constraint \cite{ThiemannQSD8}, i.e. $\hat{C}^{\dagger} \hat{C} \ket{\psi}=0$ implies $\hat{C} \ket{\psi}=0$. In the master constraint language, one subtracts an operator from the master constraint which vanishes in the classical limit to obtain a sufficiently large solution space.}, 
\be
	k^2 =  \left( \frac{2 j^-}{1-\gamma} \right)^2 =  \left( \frac{2 j^+}{1+ \gamma} \right)^2
\ee
where $k$ is the quantum number associated to the Casimir operator of the SU$(2)$ subgroup stabilising $N^I$. Depending on the value of the Barbero-Immirzi parameter, either $k = j^+ + j^-$ or $k = | j^+ - j^- |$ is selected by this constraint. The EPRL intertwiner for SO$(4)$ spin networks with arbitrary valency \cite{KaminskiSpinFoamsFor} is then constructed by first coupling the two SU$(2)$ subgroups of SO$(4)$ holonomies in the representations $(j^+,j^-)$, calculated along incoming and outgoing edges to the intertwiner, to the $k$ representation. Then, the $k$ representations associated to each edge are coupled via an SU$(2)$ intertwiner and the complete construction is then projected into the set of SO$(4)$ intertwiners. 

An alternative derivation proposed by Ding and Rovelli \cite{DingTheVolumeOperator} makes use of weakly implementing the linear simplicity constraints, i.e. restricting to a subspace $\mathcal{H}^{\text{ext}}$ such that 
\be
	 \braopket{\phi}{\hat{C}}{\psi}=0 ~~~ \forall ~ \ket{\phi}, \ket{\psi} \in \mathcal{H}^{\text{ext}} \text{.} 
\ee
In this approach, one can also show that the volume operator restricted to $\mathcal{H}^{\text{ext}}$ has the same spectrum as in the canonical theory, which is an important test to establish a relation between the canonical theory and the EPRL model. 

Closely related to what we already observed in the previous subsection on projected spin networks, the EPRL model makes heavy use of the fact that SO$(4)$ splits into two SU$(2)$ subgroups and that the Barbero-Immirzi parameter is available in four dimensions. Thus, we would have to restrict to four dimensions with non-vanishing $\gamma$ if we would want to use EPRL solution to the simplicity constraints. One upside of this solution when comparing to our proposition for solving the quadratic constraints is that no choice problem occurs, i.e. if we map the quantum numbers of the EPRL intertwiners to SU$(2)$ spin networks, a change of recoupling basis in the SU$(2)$ spin networks results again in EPRL intertwiners solving the same simplicity constraints.
The problem of stability of the solution space $\mathcal{H}^{\text{ext}}$ of the simplicity constraint under the action of the Hamiltonian constraint is however, to the best of our knowledge, not circumvented when using EPRL intertwiners. 

Also, in order to use the EPRL solution in the canonical framework, one would have to discuss exactly what it means to use linear and quadratic simplicity constraints in the same formulation, i.e. if one can freely interchange them and how continuity of the time normal field is guaranteed at the classical level if one changes from the quadratic constraints to linear constraints from one point on the spatial hypersurface to another. 
The mixed quantisation proposed in section \ref{sec:mixed} can be seen as an attempt to using both the time normal as an independent variable as well as quadratic simplicity constraints. In this case, the main difference is the presence of an additional constraint relating the time normal constructed from the generalised vielbeins to the independent time normal (which could be used in the linear simplicity constraints). 
In section \ref{sec:mixed}, this additional constraint was regularised as a master constraint which acts only on vertices. Taking the point of view that one can freely change between using the quadratic constraints plus this additional constraint {\it or} the linear constraints, one could choose the linear constraints for vertices and the quadratic constraints for edges. Since we can use a factor ordering for the master constraint where a commutator between a holonomy and a volume operator is ordered to the right, the master constraint would vanish on edges and only the quadratic simplicity constraints would have to be implemented, which are however not problematic. At vertices, we would be left with the linear constraints and could use the EPRL intertwiners. Thus, the EPRL solution seems to be a viable option in four dimensions. Whether one considers it natural or not to use both linear and quadratic constraints in the same formulation is a matter of personal taste. Nevertheless, it would be desirable to have only one kind of simplicity constraints. 

A further comment is due on the starting point of spin foam models, which is a BF-theory subject to simplicity constraints. It has been argued by Alexandrov \cite{AlexandrovNewVerticesAnd} that the secondary constraints resulting from the canonical analysis, i.e. the $D^{ab}_{\overline{M}}$-constraints on the torsion of $A_{aIJ}$ from our companion paper \cite{BTTII}, should be taken into account also in spin foam models. In the present canonical formulation, these constraints were removed by the gauge unfixing procedure \cite{BTTII} and thus do not have to be taken into account here. The requirement for the validity of this step was to modify the Hamiltonian constraint by an additional term quadratic in the $D^{ab}_{\overline{M}}$-constraints (the gauge unfixing term) which renders the simplicity constraints stable. While this ensures that we have to deal only with the non-commutativity of the (singularly smeared, or quantum) simplicity constraints in the present paper, the converse does not necessarily follow: Since the Hamiltonian constraint one obtains from the canonical analysis of BF-theory subject to simplicity constraints, the classical starting point of spin foam models, is not the modified Hamiltonian constraint considered here, but the one which results in the secondary $D^{ab}_{\overline{M}}$-constraints, it does not follow that these secondary constraints do not have to be taken into account in spin foam models. On the other hand, the present formulation hints that it might be possible to construct a spin foam model subject to simplicity constraints (and not $D^{ab}_{\overline{M}}$-constraints) which coincides with the dynamics defined by the modified Hamiltonian constraint.
In fact, it was recently shown that the transfer operator of spin foam models can be written as $T = $$\mu^\dagger W \mu$ (here for the EPRL model) \cite{DittrichHolonomySpinFoam}, where $\mu$ projects onto the solution space of the simplicity constraint. Taking into account the philosophy of spin foam models to impose the simplicity constraint at every time step in order to ensure that the second class $D^{ab}_{\overline{M}}$-constraints are satisfied, it is conceivable that the gauge unfixing term of the Hamiltonian constraint in \cite{BTTI, BTTII} could emerge from these $\mu$-projections when taking the continuum limit of the spin foam transfer operator. Thus, in the light of plausible arguments for both sides, only an explicit calculation will be able to decide this issue.

As a last remark, we point out that the non-commutativity of the linear simplicity constraints in the EPRL model results from using $\gamma \neq 0$ and thus we are not faced with this problem in higher dimensions. Essentially, as discussed in more detailed in remark \ref{rem:LinearAnomaly} of section \ref{sec:LinearRemarks}, while the rotations stabilising $N^I$ form an SO$(D)$ subgroup of SO$(D+1)$, the linear simplicity constraints in four dimensions with $\gamma \neq 0$ and $\beta \neq 0$ do not generate such a subgroup.

\section{Discussion and Conclusions}

\label{sec:Conclusion}

Let us briefly discuss the results of this paper and judge the different approaches. 

First, the mechanism for avoiding the non-commutativity in the quadratic simplicity constraints discussed in section \ref{sec:QuadraticSimplicity} is new to the best of our knowledge and we do not see any indication that the solution space is identical to previous results (up to the fact that it has the same ``size'' as SU$(2)$ spin networks). In the spin foam literature, the linear simplicity constraints are cornerstones of the new spin foam models and have been introduced since the quadratic simplicity constraints acting on vertices do not commute. While the methods for treating supergravity discussed in \cite{BTTVI} necessarily need an independent time normal and thus suggest using linear simplicity constraints, there is no need for the linear constraints in pure gravity (except for the fact that they exclude the topological sector in four dimensions). Therefore, one should not dismiss the quadratic constraints, especially since the linear constraints come with their own problems in the canonical approach. The solution presented in section \ref{sec:QuadraticSimplicity} is certainly not free of problems, most prominently the choice of the maximal commuting subset, but its close relation the SU$(2)$ based theory and the (natural) unitarity of the intertwiner map to SU$(2)$ intertwiners make it look very promising.   

The linear simplicity constraints come with their own set of problems, many of which were already known in the spin foam literature. While the results of section \ref{sec:QuadraticSimplicity} would naturally lead us to consider the quadratic constraints, the connection formulation of higher dimensional supergravity developed in \cite{BTTVI} makes it necessary to use an independent time normal as an additional phase space variable. This time normal would naturally point towards using linear simplicity constraints, although the mixed quantisation of section \ref{sec:mixed} could avoid this. Since there is no anomaly appearing when using the linear simplicity constraints (with $\gamma=0$ in four dimensions), we should implement them strongly. However, this leads to a solution space very different from the SU$(2)$ spin networks. At this point, it seems to be best to let oneself be guided by physical intuition and the results from the quadratic simplicity constraints as well as the desired resemblance to SU$(2)$ spin networks. Ad hoc methods for getting close to this goal have been discussed in section \ref{sec:LinearRemarks}. We however stress that these methods are, as said, ad hoc and they don't follow from standard quantisation procedures. The mixed quantisation discussed at the end of section \ref{sec:TheLinearSimplicity} also does not seem completely satisfactory, especially since the master constraint ensuring the equality of the independent normal $N^I$ and the derived normal $n^I(\pi)$ is very complicated to solve. Nevertheless, in section \ref{subsec:EPRL}, an application to EPRL intertwiners is outlined which could avoid this problem by using linear simplicity constraints for the vertices. The strength of the mixed quantisation is thus that it provides a mechanism to incorporate both the quadratic simplicity constraints as well as an independent time normal in the same canonical framework, which is what is done on the path integral side in the EPRL model. 

A comparison to results from the spin foam literature, especially projected spin networks and the EPRL model, shows that many of the problems connected with using the linear simplicity constraints have already been known, partly in different guises. While using these known results in our framework seems to be a viable option in four dimensions, we are unaware of possible ways to extend them also to higher dimensions since main ingredients are a non-vanishing Barbero-Immirzi parameter as well as special properties of SO$(4)$. 

In conclusion, we reported on several new ideas of how to treat the simplicity constraints
which appear in our connection formulation of general relativity in any dimension $D\ge 3$
\cite{BTTI,BTTII,BTTIII,BTTIV,BTTVI,BTTVII} and found that none of the presented ideas are entirely satisfactory at this point and further research on the open questions needs to be conducted. We hope that the discussion presented in this paper will be useful for an eventually
consistent formulation.
\\
\\
\\
{\bf\large Acknowledgements}\\
NB and AT thank Emanuele Alesci, Jonathan Engle, Alexander Stottmeister, and Antonia Zipfel for numerous discussions as well as Karl-Hermann Neeb and Toshiyuki Kobayashi for counsel on representation theory. NB and AT thank the German National Merit Foundation for financial support. The part of the research performed at the Perimeter Institute for Theoretical Physics was supported in part by funds from the Government of Canada through NSERC and from the Province of Ontario through MEDT. 
During final improvements of this work, NB was supported by the NSF Grant PHY-1205388 and the Eberly research funds of The Pennsylvania State University.
We acknowledge helpful suggestions from the referees of this paper which greatly improved its readability.  

\newpage

\begin{appendix}

\setcounter {equation} {0}

\renewcommand\theequation{\thesection .\arabic{equation}}

\section{Simple Irreps of SO$(D+1)$ and Square Integrable Functions on the Sphere $S^{D}$}
\label{app:SO(D+1)}
There is a natural action of SO$(D+1)$ on $F \in \mathcal{H} := L_2(S^D, d\mu)$ given by $\pi(g)F(N) := F(g^{-1} N)$. The $\pi(g)$ are called quasi-regular representations of SO$(D+1)$. The generators in this representation are of the form $\tau_{IJ} = \frac{1}{2} (\frac{\partial}{\partial N^I} N_J - \frac{\partial}{\partial N^J} N_I)$ and are known to satisfy the quadratic simplicity constraint $\tau_{[IJ}\tau_{KL]} = 0$ \cite{FreidelBFDescriptionOf}. These representations are reducible. The representation space can be decomposed into spaces of harmonic homogeneous polynomials $\mathfrak{H}^{D+1, l}$ of degree $l$ in $D+1$ variables, $L_2(S^D) = \sum_{l=0}^{\infty} \mathfrak{H}^{D+1, l}$. The restriction of $\pi(g)$ to these subspaces gives irreducible representations of SO$(D+1)$ with highest weight $\vec{\Lambda} = (l,0,...,0)$, $l \in \mathbb{N}$. These are (up to equivalence) the only irreducible representations of SO$(D+1)$ satisfying the quadratic simplicity constraint \cite{FreidelBFDescriptionOf} and therefore are mostly called simple representations in the Spin Foam community, which we will adopt in this work. Note that these representations have been studied quite extensively in the mathematical literature, where they are called most degenerate representations \cite{HormessMoreOnCoupling, AlisauskasCouplingCoefficientsOf, JunkerExplicitEvaluationOf}, (completely) symmetric representations \cite{AlisauskasCouplingCoefficientsOf, AlisauskasSomeCouplingAnd, GirardiGeneralizedYoungTableaux, GirardiKroneckerProductsFor} or representations of class one (with respect to a SO$(D)$ subgroup) \cite{VilenkinSpecialFunctionsAnd}.  The latter is due to the fact that these representations of SO$(D+1)$ are the only ones which have in their representations space a non-zero vector invariant under a SO$(D)$ subgroup, which is exactly the definition of being of class one w.r.t. a subgroup given in \cite{VilenkinSpecialFunctionsAnd}. An orthonormal basis in $\mathfrak{H}^{D+1, l}$ is given by generalisations of spherical harmonics to higher dimensions \cite{VilenkinSpecialFunctionsAnd} which we denote $\Xi^{\vec K}_l(N)$,
\be
\int_{S^D} \Xi^{\vec K}_l(N) ~\overline{\Xi^{\vec M}_{l'}(N)} ~ dN = \delta_{l'}^l \delta^{\vec{K}}_{\vec{M}}\text{,}
\ee
where $\vec{K}$ denotes an integer sequence $\vec{K} := (K_1, \hdots,K_{D-2}, \pm K_{D-1}) $ satisfying $l \geq K_1 \geq \hdots \geq K_{D-1} \geq 0$ and analogously defined $\vec{M}$. $F_l(N) \in \mathfrak{H}^{D+1, l}$ can be decomposed as $F_l(N) = \sum_{\vec K} a_{\vec{K}} \Xi^{\vec K}_l(N)$ where the sum runs over those integer sequences $\vec{K}$ allowed by the above inequality. Since $L_2(S^D) = \sum_{l=0}^{\infty} \mathfrak{H}^{D+1, l}$, any square integrable function $F(N)$ on the sphere can be expanded in a mean-convergent series of the form \cite{VilenkinSpecialFunctionsAnd} 
\be
F(N) = \sum_{l = 0}^{\infty} \sum_{\vec K_l} a^l_{\vec K_l} \Xi_l^{\vec K_l}(N)\text{.}
\ee

Consider a recoupling basis \cite{AgrawalaGraphicalFormulationOf} for the ONB of the tensor product of $N$ irreps: Choose a labelling of the irreps $\vec{\Lambda}_1,...,\vec{\Lambda}_N$. Then, consider the ONB 
\be
\left| \vec{\Lambda}_1,...,\vec{\Lambda}_N; \vec{\Lambda}_{12},\vec{\Lambda}_{123},...,\vec{\Lambda}_{1...N-1}; \vec{\Lambda},\vec{M} \right\rangle \text{,}
\ee
(with certain restrictions on the values of the intermediate and final highest weights).
 A basis in the intertwiner space is given by 
\be
\left| \vec{\Lambda}_1,...,\vec{\Lambda}_N; \vec{\Lambda}_{12},\vec{\Lambda}_{123},...,\vec{\Lambda}_{1...N-1}; 0,0 \right\rangle \text{,}
\ee
(with certain restrictions). A change of recoupling scheme corresponds to a change of basis in the intertwiner space. A basis in the intertwiner space of $N$ simple irreps is given by
\be
\left| \Lambda_1,...,\Lambda_N; \vec{\Lambda}_{12},\vec{\Lambda}_{123},...,\vec{\Lambda}_{1...N-1}; 0,0 \right\rangle \text{,}
\ee
(with certain restrictions), since in the tensor product of two simple irreps, non-simple irreps appear in general \cite{GirardiKroneckerProductsFor, GirardiGeneralizedYoungTableaux}, 
\be
(\lambda_1,0,...,0) \otimes (\lambda_2,0,...,0) = \sum_{k = 0}^{\lambda_2}\sum_{l=0}^{\lambda_2-k} (\lambda_1 + \lambda_2 - 2k - l,l,0,...,0) \hspace{5mm} (\lambda_2 \leq \lambda_1)\text{.}
\ee

\end{appendix}

\newpage

\bibliography{pa91pub.bbl}

\end{document}